\documentclass[runningheads]{llncs}
\usepackage{thmtools,thm-restate}
\usepackage{amsmath,amssymb,amsfonts,cases}
\usepackage{mathtools}
\usepackage[usenames,dvipsnames]{color}
\usepackage{xcolor}
\usepackage{xspace}
\usepackage{microtype}
\usepackage{todonotes}
\usepackage{colonequals}
 \usepackage{booktabs}
\usepackage{cite}
\usepackage{tcolorbox}
\tcbuselibrary{theorems}

\usepackage{twoopt}

\usepackage{multirow}
\usepackage{multicol}

\usepackage{paralist}

\usepackage{subfigure}
\usepackage{url}
\usepackage{appendix}
\usepackage{graphicx}
\usepackage{algorithm}
\usepackage[noend]{algorithmic}
\usepackage{wrapfig}
\usepackage{listings}
\usepackage{nicefrac}
\usepackage{tikz} 
\usepackage{pgfplots}
\usepackage{mdframed}
\usepackage[super]{nth}
\usepackage{ dsfont }
\usepackage{pifont}
\usepackage{etoolbox}
\usepackage{framed}
\usepackage{xcolor}
\usepackage{marvosym}
\usepackage[breaklinks=true]{hyperref}
\usepackage[capitalize]{cleveref}
\usepackage{adjustbox}
\usepackage{etoolbox}
\usepackage{framed}
\usepackage{xcolor}

\usetikzlibrary{arrows,decorations.pathmorphing,positioning,fit,trees,shapes,shadows,automata,calc,decorations.pathmorphing} 
\tikzset{outline/.style args={#1}{%
  draw=#1,thick,fill=#1!50},initial text={}, every state/.style={minimum size=2.2em}}

\algsetup{linenosize=\scriptsize}

\newcommand{\ie}{i.e.\@\xspace}

\newcommand{\eg}{e.g.\@\xspace}

\newcommand{\prophesy}{\textrm{PROPhESY}\xspace}

\newcommand{\storm}{\textrm{Storm}\xspace}

\newcommand{\dtmc}{\ensuremath{\mathcal{D}}\xspace}

\newcommand{\pDtmcInit}[1][]{\ensuremath{\pdtmc{#1}=(S{#1},\sinit{#1},T{#1},\Paramvar{#1},\probdtmc{#1})}}
\newcommand{\pDtmcInitGood}[1][]{\ensuremath{\pdtmc{#1}=(S{#1},\sinit{#1},\{\good\},\Paramvar{#1},\probdtmc{#1})}}

\newcommand{\wfa}{\ensuremath{\mathcal{A}}\xspace}
\newcommand{\wfaTrans}{\ensuremath{E}}

\newcommand{\wfaInit}[1][]{\ensuremath{\wfa{#1}=(S{#1},\sinit{#1},T{#1},\Paramvar{#1},\wfaTrans{#1})}}

\newcommand{\gd}{GD\xspace}
\newcommand{\dpmctxt}{derived weighted automaton\xspace}
\newcommand{\dpmc}{\ensuremath{\pt{} \pdtmc}\xspace}
\newcommand{\dpmcTrans}{\ensuremath{E}}

\newcommand{\dpmcInitPrime}[1][]{\ensuremath{\dpmc{#1}=(S',\pt{} \sinit, T, \Paramvar{#1},\dpmcTrans)}}

\newcommand{\pr}{\ensuremath{\mathrm{Pr}}}
\newcommand{\reachPr}[3]{\ensuremath{\pr_{#3}(#1 \models \finally #2)}}

\newcommand{\er}{\ensuremath{\mathrm{ER}}}
\newcommand{\expRew}[3]{\ensuremath{\er_{#3}(#1 \models \finally #2)}}
\newcommandtwoopt{\expRewT}[2][][]{ \ifthenelse{\equal{#2}{}}
	{\expRew{#2}{#1}{T}}
	{\expRew{#1}{#2}{T}}}
\newcommandtwoopt{\expRewGood}[2][][]{ \ifthenelse{\equal{#2}{}}
	{\expRew{#2}{#1}{\good}}
	{\expRew{#1}{#2}{\good}}}

\DeclareMathOperator{\rewFunction}{\text{\it rew}}
\newcommand{\finally}{\lozenge}

\newcommand{\colorpar}[3]{\colorbox{#1}{\parbox{#2}{#3}}}
\newcommand{\marginremark}[3]{\marginpar{\colorpar{#2}{\linewidth}{\color{#1}#3}}}

\newcommand{\R}{\mathbb{R}}

\DeclareMathOperator{\Vect}{\textit{pFun}}
\DeclareMathOperator{\Distr}{\textit{pDistr}}
\DeclareMathOperator{\wfaDistr}{\textit{pQDistr}}
\newcommand{\distDom}{X}

\newcommand{\distFunc}{\mu}
\newcommand{\distDomElem}{x}

\newcommand*{\Paramvar}{{V}} % Set of parameters
\newcommand*{\ParamSpace}{\mathcal{U}} % Parameter space / Universe. Was U^V before

\newcommand{\State}{\ensuremath{\mathbb{S}}\xspace}      
\newcommand{\sinit}{s_{\mathit{I}}} % initial state of DTMC/MDP

\newcommand{\probdtmc}{\mathcal{P}}

\newcommand{\pdtmc}{\ensuremath{\mathcal{M}}\xspace}

\renewcommand{\Pr}{\ensuremath{\textnormal{Pr}}}

\newcommand{\act}{\ensuremath{\alpha}}

\newcommand{\pathset}[2]{\Paths^{#1}_{#2}}
\newcommand{\paths}[2][]{ \ifthenelse{\equal{#1}{}}
	{\pathset{#2}{}}
	{\pathset{#2}{#1}}}

\newcommand{\bad}{\ensuremath{\lightning}}

\newcommand{\Path}{\pi}
\newcommand{\Paths}{\mbox{\sl Paths}}

\newcommand{\sol}[2]{\ensuremath{\pr_{}^{#1 \rightarrow #2}}}% I removed $M$
\newcommand{\solWithM}[3][]{\ensuremath{\pr_{#3}^{#1 \rightarrow #2}}}
\newcommand{\solRew}[2]{\solRewWithM{#1}{#2}{}}%
\newcommand{\solRewWithM}[3]{\ensuremath{\er_{#3}^{#1 \rightarrow #2}}}
\DeclareRobustCommand{\good}{\Simley{0.5}{0.2}\xspace}
\DeclareRobustCommand{\bad}{\Simley{-0.5}{0.2}\xspace}
\newcommand{\Simley}[2]{%
	\begin{tikzpicture}[scale=#2]
	\newcommand*{\SmileyRadius}{1.0}%
	;
	
	\pgfmathsetmacro{\eyeX}{0.5*\SmileyRadius*cos(30)}
	\pgfmathsetmacro{\eyeY}{0.5*\SmileyRadius*sin(30)}
	\draw [line width=0.25mm] (\eyeX-0.25,\eyeY) -- (\eyeX-0.25,\eyeY+0.375);
	\draw [line width=0.25mm] (-\eyeX+0.25,\eyeY) -- (-\eyeX+0.25,\eyeY+0.375);
	
	\pgfmathsetmacro{\xScale}{2*\eyeX/180}
	\pgfmathsetmacro{\yScale}{1.0*\eyeY}
	\draw[line width=0.25mm, domain=-\eyeX:\eyeX]
	plot ({\x},{
		-0.1+#1*0.15 % shift the smiley as smile decreases
		-#1*1.75*\yScale*(sin((\x+\eyeX)/\xScale))-\eyeY});
	\end{tikzpicture}%
}%

\DeclareMathAlphabet{\mathpzc}{OT1}{pzc}{m}{it}
\def\presuper#1#2%
  {\mathop{}%
   \mathopen{\vphantom{#2}}^{#1}%
   \kern-\scriptspace%
   #2}

\spnewtheorem{definition}{Definition}{\bfseries}{\itshape}

\spnewtheorem{defboxed}[definition]{Definition}{\bfseries}{\itshape}
\newtcolorbox{mymathbox}[1][]{colback=white, sharp corners, #1}
\tcolorboxenvironment{defboxed}{colframe=black, colback=white, sharp corners}
\usepackage{csquotes}
\usepackage{xcolor}
\definecolor{onyx}{HTML}{393E41}
\definecolor{redviolet}{HTML}{8B2635}
\definecolor{munsellblue}{HTML}{4B88A2}
\newcommand{\dervar}[2]{\partial_{#1} #2}
\newcommand{\Prob}[0]{\mathcal{P}}
\newcommand{\MC}[0]{\mathcal{M}}

\newcommand{\rv}[1]{\textcolor{redviolet}{#1}}
\newcommand{\mb}[1]{\textcolor{munsellblue}{#1}}
\newcommand{\mbdervar}[2]{\textcolor{munsellblue}{\partial_{#1} #2}}
\newcommand \LTLUntil      {\mathbin{\mathcal{U}\kern-.1em}}
\newcommand{\pt}{\partial_{p}}
\newcommand{\pti}{\partial_{p_i}}

\newcommand{\bestmethod}{\momentumsign}
\newcommand{\Bestmethod}{\momentumsign}
\newcommand{\bestrestrictionmethod}{\projectionwith}

\newcommand{\pso}{PSO\xspace}
\newcommand{\qcqp}{QCQP\xspace}
\newcommand{\easyparameters}{\enquote{easy-parameters}\xspace}

\newcommand{\momentum}{Momentum\xspace}
\newcommand{\momentumsign}{Momentum-Sign\xspace}
\newcommand{\adam}{Adam\xspace}
\newcommand{\radam}{RAdam\xspace}
\newcommand{\nag}{NAG\xspace}

\newcommand{\rmsprop}{RMSProp\xspace}
\newcommand{\plaingd}{Plain \gd}

\newcommand{\projectionwith}{projection\xspace}

\newcommand{\barfunc}{{bar}}

\definecolor{shadecolor}{named}{lightgray}
\Crefname{figure}{Fig.}{Figs.}
\crefname{figure}{Fig.}{Figs.}
\Crefname{definition}{Def.}{Defs.}
\Crefname{definition}{Def.}{Defs.}
\Crefname{tabular}{Tab.}{Tabs.}
\crefname{tabular}{Tab.}{Tabs.}
\Crefname{section}{Sec.}{Secs.}
\crefname{section}{Sec.}{Secs.}
\crefname{algocf}{Alg.}{Algs.}
\Crefname{algocf}{Alg.}{Algs.}
\Crefname{equation}{Eq.}{Eqs.}
\crefname{equation}{Eq.}{Eqs.}
\Crefname{ALC@unique}{Line}{Lines}
\crefname{ALC@unique}{line}{lines}

\makeatletter
\def\orcidID#1{\smash{\href{http://orcid.org/#1}{\protect\raisebox{-1.25pt}{\protect\includegraphics{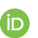}}}}}
\makeatother

\newtoggle{TR}
\toggletrue{TR}
\newtoggle{comments}
\togglefalse{comments}
\iftoggle{comments}{

\newcommand{\jm}[1]{\marginremark{white}{red!70!blue}{\scriptsize{[JM]~ #1}}}
\newcommand{\jpk}[1]{\marginremark{white}{green!70!black}{\scriptsize{[JPK]~ #1}}}
\newcommand{\js}[1]{\marginremark{white}{orange}{\scriptsize{[JS]~ #1}}}
\newcommand{\lh}[1]{\marginremark{white}{red!70!black}{\scriptsize{[LH]~ #1}}}
\newcommand{\sj}[1]{\marginremark{white}{blue}{\scriptsize{[SJ]~ #1}}}

}{
\newcommand{\jm}[1]{}
\newcommand{\jpk}[1]{}
\newcommand{\js}[1]{}
\newcommand{\lh}[1]{}
\newcommand{\sj}[1]{}
}

\title{Gradient-Descent for Randomized Controllers under Partial Observability}
\author{Linus~Heck\inst{1}\orcidID{0000-0002-4774-7609} \and Jip Spel\inst{1}$^{(\text{\Letter})}$\orcidID{0000-0002-9113-2791}  \and Sebastian~Junges\inst{2}\orcidID{0000-0003-0978-8466}  \and \\ Joshua~Moerman\inst{1,3}\orcidID{0000-0001-9819-8374} \and Joost-Pieter~Katoen\inst{1}\orcidID{0000-0002-6143-1926}}

\authorrunning{L.~Heck, J.~Spel, S.~Junges, J.~Moerman, J.-P.~Katoen}
\institute{
RWTH Aachen University, Aachen, Germany
 \thanks{Supported by DFG RTG 2236 ``UnRAVeL'' and ERC AdG 787914 FRAPPANT.}
\and
Radboud University, Nijmegen, the Netherlands
\and Open University of the Netherlands, Heerlen, the Netherlands
}
\tikzset{every state/.style={minimum size=2.2em}}

\begin{document}
\lstset{mathescape=true, tabsize=2}

\maketitle
\begin{abstract}
Randomization is a powerful technique to create robust controllers, in particular in partially observable settings. 
The degrees of randomization have a significant impact on the system performance, yet they are intricate to get right. 
The use of synthesis algorithms for parametric Markov chains (pMCs) is a promising direction to support the design process of such controllers. 
This paper shows how to define and evaluate gradients of pMCs. 
Furthermore, it investigates varieties of gradient descent techniques from the machine learning community to synthesize the probabilities in a pMC.
The resulting method scales to significantly larger pMCs than before and empirically outperforms the state-of-the-art, often by at least one order of magnitude.
\end{abstract}\js{artifact badge}

\section{Introduction}
Markov chains (MCs) are the common operational model to describe closed-loop systems with probabilistic behavior, i.e., systems together with their controllers whose behavior is described by a stochastic process (Fig.~\ref{fig:closedloop}). 
Examples include self-stabilizing protocols for distributed systems~\cite{DBLP:conf/podc/IsraeliJ90} and exponential back-off mechanisms in wireless networks. 
Randomization is also important for robustness in autonomous systems with noisy sensors~\cite{DBLP:books/daglib/0014221}, obfuscation and (fuzz) test-coverage~\cite{DBLP:conf/cav/FremontS18}. 
Such systems are typically subject to temporal specifications, e.g., with high probability an autonomous system should not crash, and a self-stabilizing protocol should reach a stable configuration in few expected steps. 
Checking system models against these specifications can be efficiently done using state-of-the-art probabilistic model checking~\cite{DBLP:conf/cav/KwiatkowskaNP11, DBLP:journals/corr/HenselJKQV20}. 
We highlight that while controllers for these systems operate under partial information, the analysis of a system with controller does not need to take partial observability into account. 

\begin{figure}
\centering
\subfigure[Verification of closed-loop systems. Memory of the controller is part of the system.]{
\scalebox{0.75}{
\begin{tikzpicture}
	\node[rectangle,draw,minimum height=0.6cm] (contr) {\scriptsize{Controller}};
	\node[rectangle,draw,right=0.6cm of contr,minimum height=0.6cm] (env) {\scriptsize{Environment}};
	\draw[->] (env.south) -- +(0,-0.2) -| node[pos=0.2,below] (txtin) {\scriptsize{observation}} (contr.south);
	\draw[->] (contr.north) -- +(0,0.2) -| node[pos=0.2,above] (txtout) {\scriptsize{fixed probability for action}} (env.north);
	\node[fit=(contr)(env)(txtout)(txtin),draw,dotted, inner sep=5pt,label={north:system = Markov chain}] (sys) {};
\end{tikzpicture}
\label{fig:closedloop}
}	
}
~
\subfigure[Synthesis of controllers. Memory not fixed and thus not part of the system.]{
\scalebox{0.75}{
\begin{tikzpicture}
	\node[rectangle,draw,minimum height=0.6cm] (contr) {\scriptsize{Controller}};
	\node[rectangle,draw,right=0.6cm of contr,minimum height=0.6cm] (env) {\scriptsize{Environment}};
	\draw[->] (env.south) -- +(0,-0.2) -| node[pos=0.2,below] (txtin) {\scriptsize{observation}} (contr.south);
	\draw[->] (contr.north) -- +(0,0.2) -| node[pos=0.3,above] (txtout) {\scriptsize{which action}?} (env.north);
	\draw[->] (contr.110) -- +(0,0.2) -| node[above,align=center] {\raisebox{-6pt}{\scriptsize{which}}\\\scriptsize{update}?} (contr.130);
	\node[fit=(contr)(env)(txtout)(txtin),dotted,draw,dotted, inner sep=5pt,label={north:system = e.g.~POMDP}] (sys) {};
\end{tikzpicture}	
\label{fig:pomdps}
}
}
~
\subfigure[Parameter synthesis for controllers. Memory fixed and part of the system.]{
\scalebox{0.75}{
\begin{tikzpicture}
	\node[rectangle,draw,minimum height=0.6cm] (contr) {\scriptsize{Controller}};
	\node[rectangle,draw,right=0.6cm of contr,minimum height=0.6cm] (env) {\scriptsize{Environment}};
	\draw[->] (env.south) -- +(0,-0.2) -| node[pos=0.2,below] (txtin) {\scriptsize{observation}} (contr.south);
		\draw[->] (contr.north) -- +(0,0.2) -| node[pos=0.25,above] (txtout) {\scriptsize{\textbf{which} probability for action?}} (env.north);
			\node[fit=(contr)(env)(txtout)(txtin),draw,dotted, inner sep=5pt,label={north:system = pMC}] (sys) {};
	\node[fit=(contr)(env),dotted] (sys) {};
\end{tikzpicture}	
\label{fig:psynt}
}
}
\caption{Verification and (syntax-guided) synthesis for controllers}
\end{figure}
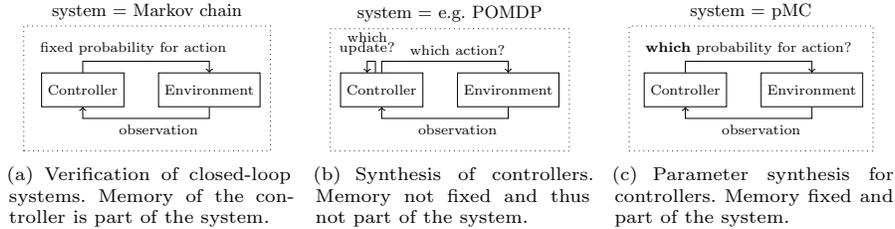

One step beyond verification is the correct-by-construction synthesis of controllers for such systems via Partially Observable Markov Decision Processes (POMDPs)~(Fig.~\ref{fig:pomdps}). 
In general, the synthesis for partial-information controllers is undecidable~\cite{DBLP:journals/ai/MadaniHC03,DBLP:conf/formats/GiroD07,DBLP:journals/jacm/BaierGB12}.
Syntax-guided synthesis~\cite{sygus} takes a simpler perspective and synthesizes only particular system aspects starting from a user-provided template.
In this paper, we focus on being provided with a template controller with a fixed memory structure (influencing the number of indistinguishable states) and a fixed set of potential actions that we want to randomize over. This setting is useful, as in many systems one randomizes on purpose, e.g., in distributed protocols to break symmetry or for robustness. 
In particular, the randomization is controllable, but selecting a (near-)optimal way to randomize is non-trivial. 

The synthesis task reduces to randomize appropriately in a system with a fixed topology (Fig.~\ref{fig:psynt}).
In this context, a controller selects a fixed set of actions (of the POMDP) $\act_1, \hdots, \act_n$ with probabilities $p_1, \hdots ,p_n$.
The aim is to synthesize a \emph{realizable} controller, that is, the result of the synthesis should not enforce to randomize differently in indistinguishable states --- such a controller depends on information which is not available at runtime and therefore cannot be implemented.
Consequently, for indistinguishable states, a realizable controller must take an action $\act$ with the same probability $p_i$.
Synthesizing such controllers can be formally described~\cite{DBLP:conf/uai/Junges0WQWK018} as feasibility synthesis in \emph{parametric} Markov chains (pMCs), i.e., MCs with symbolic probabilities $p_1, \hdots, p_n$~\cite{Daws04,DBLP:journals/fac/LanotteMT07}. 
The feasibility synthesis task asks to find values $u_1, \hdots, u_n$ for the parameters such that the MC satisfies a given property. 
This problem has been studied extensively in the literature, \eg in~\cite{DBLP:conf/tase/ChenHHKQ013, DBLP:conf/tacas/Cubuktepe0JKPPT17,DBLP:conf/atva/GainerHS18, DBLP:conf/nfm/HahnHZ11, DBLP:conf/atva/QuatmannD0JK16}, see also the related work section.

\begin{example}
	\Cref{fig:POMDPtoAnyPMC:POMDP} depicts a POMDP.
	The colors match the observations at a state.
	When observing a red state, $s_1$ or $s_3$, with probability $q_1$ action $\act_1$ is taken and with probability $q_2$ action $\act_2$. 
	At state $s_0$ action $\act_i$ is taken with probability $p_i$. 
	This directly results in the pMC of~\cref{fig:POMDPtoAnyPMC:PMC}.
\end{example}

\begin{figure}[t]
	\centering
	\subfigure[POMDP]{
		\centering
		\scalebox{0.8}{
			\begin{tikzpicture}[every node/.style={font=\scriptsize},st/.style={circle, inner sep=2pt, draw}]
\node[st, fill=blue!40, initial, initial text=] (s0) {$s_0$};
\node[st, right=3cm of s0, fill=white!40] (s2)  {$s_2$};
\node[st, above=of s2, fill=red!60] (s1)  {$s_1$};
\node[st, below=of s2, fill=red!60] (s3)  {$s_3$};

\node[circle, inner sep=1.5pt, fill=black, left=1.5cm of s1] (a1) {};
\node[circle, inner sep=1.5pt, fill=black, left=1.5cm of s2] (a2) {};
\node[circle, inner sep=1.5pt, fill=black, left=1.5cm of s3] (a3) {};
\node[circle, inner sep=1.5pt, fill=black, right=0.75cm of s1] (a4) {};
\node[circle, inner sep=1.5pt, fill=black, below=0.4cm of s1] (a5) {};
\node[circle, inner sep=1.5pt, fill=black, right=0.75cm of s3] (a6) {};
\node[circle, inner sep=1.5pt, fill=black, above=0.4cm of s3] (a7) {};

\draw[->] (s0) -- node[above, xshift=-0.1cm] {$\act_1$} (a1);
\draw[->] (s0) -- node[below] {$\act_2$} (a2);
\draw[->] (s0) -- node[below, xshift=-0.1cm] {$\act_3$} (a3);
\draw[->] (s1) -- node[above] {$\act_1$} (a4);
\draw[->] (s1) -- node[left] {$\act_2$} (a5);
\draw[->] (s3) -- node[below] {$\act_1$} (a6);
\draw[->] (s3) -- node[left] {$\act_2$} (a7);

\draw[->] (a1) -- node[above] {$1$} (s1);
\draw[->] (a2) -- node[above] {$\nicefrac{1}{2}$} (s2);
\draw[->] (a2) -- node[left] {$\nicefrac{1}{2}$} (s3);
\draw[->] (a3) -- node[above] {$1$} (s3);

\draw[->] (a4) edge[bend left] node[above] {$1$} (s2);

\draw[->] (a5) -- node[above] {$\nicefrac{1}{2}$} (s0);
\draw[->] (a5) -- node[left] {$\nicefrac{1}{2}$} (s2);

\draw[->] (a6) edge[bend right] node[above] {$1$} (s2);

\draw[->] (a7) edge[bend left] node[right] {$1$} (s3);
\end{tikzpicture}
		}
		\label{fig:POMDPtoAnyPMC:POMDP}
	}\hspace{2em}
	\subfigure[pMC]{
		\centering
		\scalebox{0.8}{
			\begin{tikzpicture}[every node/.style={font=\scriptsize},st/.style={circle, inner sep=2pt, draw}]
\node[st,initial,initial text=] (s0) {$s_0$};
\node[st, right=3cm of s0] (s2)  {$s_2$};
\node[st, above=of s2] (s1)  {$s_1$};
\node[st, below=of s2] (s3)  {$s_3$};

\draw[->] (s0) -- node[above] {${\color{blue!40}p_1} \cdot 1$} (s1);
\draw[->] (s0) -- node[above] {${\color{blue!40}p_2} \cdot \nicefrac{1}{2}$} (s2);
\draw[->] (s0) -- node[below, xshift=-0.4cm] {${\color{blue!40}p_3 + p_2} \cdot \nicefrac{1}{2}  $} (s3);

\draw[->] (s3) edge[loop right] node[right] {$\color{red!60}q_2$} (s3);
\draw[->] (s3) edge node[right] {$\color{red!60}q_1$} (s2);
\draw[->] (s1) edge[bend right] node[left] {$\nicefrac{1}{2}\cdot {\color{red!60}q_2}$\ \ } (s0);
\draw[->] (s1) edge node[right] {$ {\color{red!60}q_1} + \nicefrac{1}{2}\cdot {\color{red!60}q_2}$} (s2);
\draw[->] (s2) edge[loop right] node[right] {$1$} (s2);
\end{tikzpicture}
		}
		\label{fig:POMDPtoAnyPMC:PMC}
	}
	\caption{From POMDPs to pMCs~\cite[p.~182]{DBLP:phd/dnb/Junges20}.}
	\label{fig:POMDPtoAnyPMC}
\end{figure}
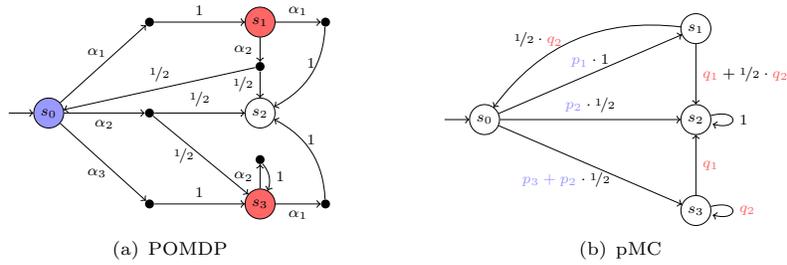

The challenge in applying parameter synthesis is twofold: whereas the problem is ETR-complete\footnote{ETR = Existential Theory of the Reals. ETR-complete decision problems are as hard as finding the roots of a multivariate polynomials.}~\cite{DBLP:journals/jcss/JungesK0W21}, the number of parameters grows linear in the number of different observations and the number of actions available to the controller. 
For many real-life applications we must thus deal with thousands of parameters. 
This scale is out of reach for exact or complete methods~\cite{DBLP:conf/cav/DehnertJJCVBKA15}.
Heuristic methods have shown some promise. 
These methods either rely on efficient model checking but are heavily sample-inefficient~\cite{DBLP:conf/tase/ChenHHKQ013}, or rely on the efficiency of convex solvers to search the parameter space in a more principled way~\cite{DBLP:conf/atva/CubuktepeJJKT18}. 

This paper presents a novel method that advances the state-of-the-art in feasibility synthesis often by one or more orders of magnitude.
The method is rooted in two key observations:
\begin{itemize}
    \item gradient-based search methods, i.e., variants of gradient search, scale to high-dimensional search spaces, and
	\item in pMCs, the gradient at a parameter evaluation can be efficiently evaluated.
\end{itemize}
In this paper, we show a principled way to evaluate gradients in parametric MCs. 
We characterize gradients as solutions of a linear equation system over the field over rational functions and alternatively as expected rewards of an automaton that is easily derived from the pMC at hand. 
Using the efficient computation of gradients, we evaluate both classical (\plaingd, \momentum \gd~\cite{DBLP:books/lib/Rumelhart89}, and Nesterov accelerated \gd~\cite{nesterov1983method, DBLP:conf/icml/SutskeverMDH13}) and adaptive  (\rmsprop~\cite{RMSProp}, \adam~\cite{DBLP:journals/corr/KingmaB14}, and \radam~\cite{DBLP:conf/iclr/LiuJHCLG020}) gradient descent methods.
We also consider the classical gradient descent methods where we only respect the sign of the gradient. 
Furthermore, we investigate various methods (projection, barrier function, logistic function) to deal with restrictions on the parameter space (e.g. parameters should represent probabilities).
Using an empirical evaluation, we show that 1) \bestrestrictionmethod outperforms the other restriction methods, 2) \momentumsign outperforms the other gradient descent methods, and 3) \bestmethod often outperforms state-of-the-art methods \qcqp and \pso.
Moreover, we discuss some domain-specific properties and the consequences for gradient descent.

We formalize our problem statement in \cref{subsec:problemstatement}, discuss the evaluation of gradient in \cref{sec:theory}, consider the use of gradient descent in \cref{sec:gd}, give an empirical evaluation in \cref{sec:empirical}, and discuss related work in \cref{sec:related}.
\cref{sec:conclusion} concludes and provides pointers for future work.

\newpage
\section{Preliminaries}
\label{sec:preliminaries}

\subsection{Parametric Markov Chains}
Let $\Paramvar$ be a set of $n$ real-valued \emph{parameters} (or \emph{variables}) $p_1,\hdots,p_n$.
Let $\R[V]$ denote the set of multivariate polynomials over $\Paramvar$.
%\sj{Polynomials should have rationals as coefficients, even if we consider real-valued solutions}
%\jm{Why? It makes the paper kinda inconsistent. Maybe you care about decidability and so on, but that is not very important here. For readability, I would like to stick to $\R$ only. (Also: we should avoid things like $\Q[V] \cup \R$ which is not even a ring.)}

A \emph{parameter instantiation} is a function $u \colon \Paramvar \to \R$.
We often denote $u$ as a vector $\vec{u} \in \R^n$ by ordering the set of variables $\Paramvar = \{ p_1, \ldots, p_n \}$ and setting $u_i = u(p_i)$.
We assume that all parameters are bounded, \ie, $\text{\it lb}_i \leq u(p_i) \leq \text{\it ub}_i$ for each parameter $p_i$. Let $R_i=[\text{\it lb}_i, \text{\it ub}_i]$ denote the bounds for parameter $p_i$ in region $R$.
The \emph{parameter space} of $V$, denoted $\ParamSpace \subseteq \R^\Paramvar$, is the set of all possible parameter values, \ie the hyper-rectangle spanned by the intervals $[\text{\it lb}_i, \text{\it ub}_i]$.
A set $R \subseteq \ParamSpace$ of instantiations is called a \emph{region}.

A polynomial $f$ can be interpreted as a function $f \colon \R^n \to \R$ where $f(u)$ is obtained by substitution, \ie in $f(u)$ each occurrence of $p_i$ in $f$ is replaced by $u(p_i)$.
To make clear where substitution occurs, we write $f[u]$ instead of $f(u)$ from now on.
We let $\pt{} f$ denote the partial derivative of $f$ with respect to $p$.

Let X be any set and let $\Vect(X) = \{ f \mid f \colon X \to \R[V] \}$ denote the set of generalized functions.
Now, let ${\Distr(\distDom) \subset \Vect(X)}$ denote the set of \emph{parametric probability distributions} over $\distDom$, \ie, the set of functions $\distFunc \colon \distDom \to \R[V]$ such that $0 \leq \distFunc(\distDomElem)[u] \leq 1$ and $\sum_{\distDomElem\in\distDom}\distFunc(\distDomElem)[u]=1$ for all $u$ in the parameter space $\ParamSpace$.

\begin{definition}\label{def:pmc}
	A \emph{parametric Markov chain (pMC)} is a tuple $\pDtmcInit$ with a finite set $S$ of \emph{states}, an \emph{initial state} $\sinit \in S$, a finite set $T \subseteq S$ of \emph{target states}, a finite set $\Paramvar$ of real-valued variables \emph{(parameters)} and a \emph{transition function} $\probdtmc \colon S \to \Distr(S)$.
\end{definition}

The parametric probability of going from state $s$ to $t$, denoted $\probdtmc(s, t)$, is given by $\probdtmc(s)(t)$. 
A pMC with $\Paramvar = \emptyset$ is a \emph{Markov chain} (MC).
We will use $\pdtmc$ to range over pMCs and $\dtmc$ to range over MCs.
Applying an instantiation $u$ to a pMC $\pdtmc$ yields MC $\pdtmc[u]$ by replacing each transition $f \in \R[\Paramvar]$ in $\pdtmc$ by $f[u]$.
An instantiation $u$ is \emph{graph-preserving} (for $\pdtmc$) if the topology of $\pdtmc$ is preserved, \ie, $\probdtmc(s,s') \neq 0$ implies $\probdtmc(s,s')[u] \neq 0$ for all states $s, s'$.
A region $R$ is graph-preserving if all $u \in R$ are graph-preserving.

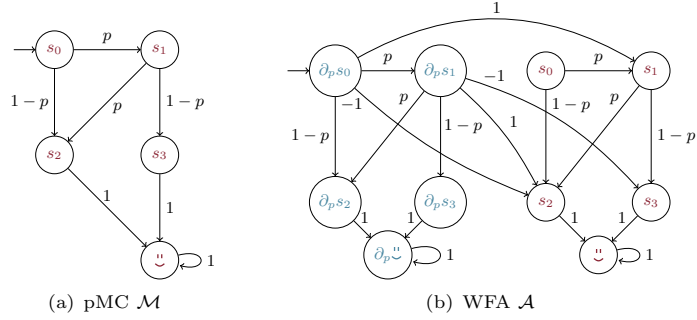
\begin{figure}[t]
	\centering	
	\subfigure[pMC \pdtmc]{
		\label{fig:pmc}
		\scalebox{0.7}{	
			\begin{tikzpicture}	[every node/.style={font=\footnotesize, node distance=2cm},simpstate/.style={circle,draw,fill=white,inner sep=1pt,minimum width=14pt}]		
				\node[state, initial] (s0) {$\rv{s_0}$};
				\node[state, right of=s0] (s1) {$\rv{s_1}$};
				\node[state, below of=s1] (s3) {$\rv{s_3}$};
				\node[state, below of=s3] (good) {$\rv{\good}$};
				\node[state, below of=s0] (s2) {$\rv{s_2}$};
				
				\draw[->](s0) edge node[auto] {$p$} (s1)
				(s0) edge node[left] {$1-p$} (s2)
				(s1) edge node[right] {$1-p$} (s3)
				(s1) edge node[right, yshift=-0.1cm] {$p$} (s2)
				(s3) edge node[right] {$1$} (good)
				(s2) edge node[above] {$1$} (good)
				(good) edge[loop right] node[right] {$1$} (good)
				;	
		\end{tikzpicture}}
}~
\subfigure[WFA \wfa]{
	\label{fig:wfa}
	\label{fig:dpmc}
	\scalebox{0.7}{	
			\begin{tikzpicture}	[every node/.style={font=\footnotesize, node distance=2cm},simpstate/.style={circle,draw,fill=white,inner sep=1pt,minimum width=14pt}]	
				\node[state, initial] (ds0) {$\mbdervar{p}{s_0}$};
				\node[state, right of=ds0] (ds1) {$\mbdervar{p}{s_1}$};
				\node[state, below of=ds1, yshift=-0.5cm] (ds3) {$\mbdervar{p}{s_3}$};
				\node[state, below of=ds3, yshift=1cm, xshift=-1cm] (dgood) {$\mbdervar{p}{\good}$};
				\node[state, below of=ds0, yshift=-0.5cm] (ds2) {$\mbdervar{p}{s_2}$};	
				\node[state, right of=ds1] (s0) {$\rv{s_0}$};
				\node[state, right of=s0] (s1) {$\rv{s_1}$};
				\node[state, below of=s1, yshift=-0.5cm] (s3) {$\rv{s_3}$};
				\node[state, below of=s3, yshift=1cm, xshift=-1cm] (good) {$\rv{\good}$};
				\node[state, below of=s0, yshift=-0.5cm] (s2) {$\rv{s_2}$};

				\draw[->](s0) edge node[auto] {$p$} (s1)
				(s0) edge node[right, yshift=0.6cm] {$1-p$} (s2)
				(s1) edge node[right] {$1-p$} (s3)
				(s1) edge node[above, yshift=0.5cm, xshift=0.3cm] {$p$} (s2)
				(s3) edge node[above, xshift=-0.05cm] {$1$} (good)
				(s2) edge node[above, xshift=0.05cm] {$1$} (good)
				(good) edge[loop right] node[right] {$1$} (good)
				% derivative
				(ds0) edge node[auto] {$p$} (ds1)
				(ds0) edge node[left] {$1-p$} (ds2)
				(ds1) edge node[right, yshift=0.2cm, xshift=-0.05cm] {$1-p$} (ds3)
				(ds1) edge node[above, yshift=0.5cm, xshift=0.3cm] {$p$} (ds2)
				(ds3) edge node[above, xshift=-0.05cm] {$1$} (dgood)
				(ds2) edge node[above, xshift=0.05cm] {$1$} (dgood)
				(dgood) edge[loop right] node[right] {$1$} (dgood)
				% between
				(ds0) edge[bend left] node[above] {$1$} (s1)
				(ds0) edge[bend right=10] node[above, xshift=-1.6cm, yshift=0.65cm] {$-1$} (s2)
				(ds1) edge[bend left=15] node[above, xshift=-1.25cm, yshift=0.6cm] {$-1$} (s3)
				(ds1) edge[bend left=10] node[auto] {$1$} (s2)
				;
			\end{tikzpicture}
		}
}
\caption{A (left) sample parametric MC and (right) its derived weighted automaton}
\end{figure}
\begin{example}
	\Cref{fig:pmc} depicts pMC \pdtmc\ with a single parameter $p$.
	Region $R = [0.1, 0.9]$ is graph-preserving, while $R=[0,0.9]$ is not graph-preserving.
\end{example}

We fix an MC $\dtmc$. 
Let $\Paths(s)$ denote the set of all infinite paths in $\dtmc$ starting from $s$, \ie, infinite sequences of the form $s_0 s_1 s_2 \ldots$ with $s_0 = s$ and $\probdtmc(s_i,s_{i+1}) > 0$.
A probability measure $\Pr_\dtmc$ is defined on measurable sets of infinite paths using a standard cylinder construction; for details, we refer to, e.g.,~\cite[Ch.~10]{BK08}.
For $T \subseteq S$ and $s \in S$, let
\begin{equation}\label{eqn:probmeasure}
\reachPr{s}{T}{\dtmc} \ = \ \Pr_\dtmc \{ \, s_0 s_1 s_2 \ldots \in \Paths(s) \mid \exists i. \, s_i \in T  \, \}
\end{equation}
denote the probability to eventually reach some state in $T$ from $s$.
For a pMC $\pdtmc$, the reachability probability depends on the parameters and so we define it as a function $\solWithM[s]{T}{\pdtmc} \colon \ParamSpace \to [0,1]$ given by $\solWithM[s]{T}{\pdtmc}[u] = \reachPr{s}{T}{{\pdtmc[u]}}$~\cite{Daws04}.
For conciseness we typically omit the subscript $\pdtmc$ and write $\solWithM[s]{T}{}$.
Zero and one reachability probabilities are preserved for graph-preserving instantiations, \ie, for all graph-preserving $u, u' \in \ParamSpace$, we have $\sol{s}{T}[u]=0$ implies $\sol{s}{T}[u']=0$ and analogously for $= 1$.
In these cases, we just write $\sol{s}{T} = 0$ or $=1$.
Let $\bad$ denote all states $s\in S$ with $\sol{s}{T} = 0$.
W.l.o.g., we assume that there is at most one $\bad$ state (this is standard preprocessing~\cite[Ch.\ 10]{BK08}).
Furthermore, we merge all states $s \in T$ into a single \good state.

\begin{example}
	For all states $s\in S$ in pMC \pdtmc from \cref{fig:pmc},
	we have $\sol{s}{\good} = 1$.
	Therefore, the pMC \pdtmc has no $\bad$ state.
\end{example}

%For a pMC $\pdtmc$, we define $\solWithM[s]{T}{\pdtmc}$ to be a function $\solWithM[s]{T}{\pdtmc} \colon \ParamSpace \to [0,1]$, with $\solWithM[s]{T}{\pdtmc}(u) = \reachPrT[s][{\pdtmc[u]}]$.
%On a graph-preserving region, the function $\solWithM[s]{T}{\pdtmc}$ is continuously differentiable~\cite{DBLP:conf/atva/QuatmannD0JK16} and admits a closed-form as a rational function over $\Paramvar$~\cite{Daws04}.
%
%In this paper we are concerned with the more general notion of \emph{expected rewards}.

\subsection{Expected Rewards}
\label{subsec:properties}
We are not only concerned with reachability probabilities but also with expected rewards.
Let \emph{state} reward function \(\rewFunction \colon S \to \R\) associate a reward to each state.
The cumulative reward for a finite path \(\hat{\Path} =  s_0 s_1 \hdots s_n\) is defined by:
\[\rewFunction(\hat{\Path}) = \rewFunction(s_0) + \rewFunction(s_1) + \hdots + \rewFunction(s_{n-1}).\]
For infinite paths $\Path = s_0 s_1 s_2 \cdots$ the reward to eventually reach $\good$ in $\pdtmc$ is:
\begin{align*}
\rewFunction(\Path, \finally \good) &= \begin{cases}
	\rewFunction(s_0 s_1 \hdots s_{n}) & \text{ if } s_i \neq \good \text{ for } 0 \leq i < n \text{ and } s_n = \good\\
	\infty & \text{ if } \Path \not \models \finally \good.
\end{cases}
\end{align*}
%As we only consider pMCs where \(\reachPr{s}{\good}{} = 1\), the latter case will not occur.

\begin{remark}
	For the sake of simplicity, we restrict ourselves to constant rewards.
	However, all notions and concepts considered in the remainder of this paper can be generalized to parametric reward functions in a straightforward manner.
	%In fact, our implementation supports parametric rewards.
\end{remark}
\begin{remark}
	From now on, we only consider graph-preserving regions and we restrict ourselves to pMCs where every state $s$ eventually reaches $\good$ almost surely, \ie, \(\sol{s}{\good} = 1\).
\end{remark}

\begin{definition}[Expected reward]
	\label{def:exprew}
	The \emph{expected reward} until reaching $\good$ from $s \in S$ for an MC \dtmc is defined as follows:
		\[
		\expRew{s}{\good}{\dtmc} \ = \ \int^{\Paths(s)}_{\Path \models \finally \good} \rewFunction(\pi, \finally \good) \cdot \Pr(\Path).
		\]
\end{definition}
\sj{Something just like this?}
%%
%\begin{definition} %% [expected reward of a state]
%Let \pdtmc a pMC with state space $S$, target state $\good \in S$, and \(\rewFunction \colon S \to \R\) be a state reward function.
%The \emph{expected reward of state $s \in S$} to reach $\good$, denoted $\phi(s)\colon \ParamSpace \to \R$, is defined by:
%	\begin{align*}
%		%\phi(\good) &= \rewFunction(\good) \\
%		\phi(s) &= \rewFunction(s) + \sum_{s' \in S} \probdtmc(s, s') \cdot \phi(s') & \text{for } s \in S \setminus \{ \good \}
%	\end{align*}
%	where we interpret $\probdtmc(s, s')$ as a function and ${+}, {\cdot}$ are pointwise operations.
%	\jm{Is this really a definition? It looks (co-)recursive... Initially I wrote something else, also $\phi$ is a bad name for a definition.}
%	\jpk{This is not a definition but a recursive characterisation. Note that the base case is missing. I suggest to first define the expected reward as a solution function, in the same way as you define transition probabilities. Then you can provide the equation system and claim that its unique solution is exactly the polynomial for the expected reward to reach $\good$.}
%\end{definition}
%%
The expected reward for a pMC $\pdtmc$ is defined analogously, but as a function $\solRewWithM{s}{\good}{\pdtmc} \colon \ParamSpace \to \R$,
given by $\solRewWithM{s}{\good}{\pdtmc}[u] = \expRew{s}{\good}{{\pdtmc[u]}}$.
Again, for conciseness we typically omit the subscript $\pdtmc$.

% we define the \emph{cumulative reward} for a finite path $\pi = s_0s_1\cdots s_n$ as $\rewFunction(\pi) = \rewFunction(s_0) + \rewFunction(s_1) + \cdots + \rewFunction(s_{n-1})$. 
%The \emph{expected reward} for eventually reaching $\good$ from $s$ is defined as \(\expRewGood[s][\pdtmc][\vec{u}] = \sum_{\pi \in \paths{\good}} \Pr(\pi)[\vec{u}] \cdot \rewFunction(\pi)\).
\begin{example}
	\label{ex:pmc:er}
	Reconsider the pMC \pdtmc from \cref{fig:pmc} with a
	state reward function $rew(s_i) = i$ for $s_i \in S \setminus \{\good\}$.
	The expected reward function $\solRew{s_0}{\good}$ is given by
	$3 \cdot p^2 + 4 \cdot p \cdot (1{-}p) + 2 \cdot (1{-}p) = {-}p^2+2 \cdot p+2$.
\end{example}
On a graph-preserving region, the function $\solRew{s}{\good}$ is always continuously differentiable~\cite{DBLP:conf/atva/QuatmannD0JK16} and admits a closed-form as a rational function over $\Paramvar$~\cite{Daws04,DBLP:conf/spin/HahnHZ09}.

\begin{remark}
Reachability probabilities are obtained by using expected rewards by letting $\rewFunction(s) = 0$ for $s \in S \setminus\{\good\}$ and $\rewFunction(\good) = 1$.
We add one sink state $s'$ s.t. $\probdtmc(s,s') = 0$ if $s \in S\setminus\{\good, \bad\}$ and  $\probdtmc(s,s') = 1$ otherwise.
The quantity $\solRew{s_0}{s'}$ now equals the reachability probability of eventually reaching $\good$.
\end{remark}

%In the parametric case, we extend this to a function which takes an instantiation.
%We allow the state reward function to be parametric: $\rewFunction \colon S \to \R[\Paramvar]$.
%The parametric expected reward is the function $\solRewWithM{s}{T}{\pdtmc}$, defined by  $\solRewWithM{s}{T}{\pdtmc}(u) = \expRewT[s][{\pdtmc[u]}]$.
%Note that for strict positive rewards $\expRewT[{\dtmc[\vec{u}]}] = \infty$ if $\reachPrT[{\dtmc[\vec{u}]}] < 1$~\cite{BK08}.

%Well-defined and graph-preserving instantiations $\vec{u}, \vec{u}'$ have the same one probabilities, \ie, $\sol[s]{T}{\pdtmc}(\vec{u}) = 1$ implies $\sol[s]{T}{\pdtmc}(\vec{u}') = 1$%, and analogous for ${=}1$.
%Consequently, we simply write $\sol[s]{T}{\pdtmc} = 1$ as it is independent from the instantiation. %(or ${=}1$)
%Let $\good$ and $\bad$ denote all states $s \in S$ with respectively $\sol[s]{T}{\pdtmc} = 1$ and $\sol[s]{T}{\pdtmc} = 0$.

%\begin{itemize}
%	\item \emph{Unbounded probabilistic reachability:} Find the instantiation $\vec{u}^*$ such that
%	\[
%	\vec{u}^* \, = \, \arg\max_{\vec{u} \in R} \, \reachPrT[{\dtmc[\vec{u}]}]
%	\]
%	\item \emph{Expected reward until a target:} Find the instantiation $\vec{u}^*$ such that
%	\[
%	\vec{u}^* \, = \, \arg\max_{\vec{u} \in R} \, \expRewT[{\dtmc[\vec{u}]}]
%	\]
%\end{itemize}
\subsection{Problem Statement}
\label{subsec:problemstatement}

This paper is concerned with the question of synthesising a randomized controller under partial observability. 
Synthesizing these controllers can formally be described~\cite{DBLP:conf/uai/Junges0WQWK018} as feasibility synthesis in pMCs.
Therefore, we consider the following question on the expected reward of eventually reaching a target state $\good$ in a given pMC $\pdtmc$ and a graph-preserving region\footnote{Technically, we use graph-preserving to ensure continuously differentiability of $\solRewWithM{s}{\good}{\pdtmc}$. For acyclic pMCs, these functions are continuously differentiable without assuming graph-preservation~\cite{DBLP:journals/jcss/JungesK0W21}.} $R$:
\begin{mdframed}[backgroundcolor=blue!5, nobreak=true, skipabove=4pt, skipbelow=0pt]
	\noindent
	Given $\lambda \geq 0$, and comparison operator $\sim$, find an instantiation $u \in R$ with:
	\[
	\expRew{s}{\{\good\}}{\pdtmc[u]} \ \sim \ \lambda.
	\]
\end{mdframed}
To solve this problem, we first show how to compute the derivative of \(\solRew{s}{\good}\) and introduce a \emph{derived} weighted automaton.
Then, we exploit this derivative by considering several gradient descent methods and applying them to solve our problem.
Finally, we show how our approach experimentally compares to existing methods from \cite{DBLP:conf/atva/CubuktepeJJKT18,DBLP:conf/tase/ChenHHKQ013}.

\section{Computing Gradients for Expected Rewards}
\label{sec:theory}
In this section, we show that we can efficiently evaluate the gradient of the function $\solRew{s}{\good}$ with respect to a parameter $p$ at an instantiation $u$.
We note that first computing $\solRew{s}{\good}$ and deriving this function symbolically is intractable: the function can be exponentially large in the number of parameters~\cite{DBLP:journals/iandc/BaierHHJKK20}.
A tractable construction follows from taking the derivative of the equation system that characterizes the expected reward~\cite[Ch.~10]{BK08}.
Alternatively, it can be obtained as an equation system for the expected rewards of a ``derived'' pMC.
Let $\pDtmcInitGood$ with reward function $\rewFunction$ and parameter $p \in V$.

\subsection{Equation-System Based Characterisation}
\label{subsec:equationsystem}

\begin{restatable}{definition}{defsoed}\label{def:soed}
The system of equations for the partial derivative of $\solRewWithM{s}{\good}{\pdtmc}$ w.r.t. $p \in V$ is given by:
	\begin{align*}
		\rv{x_s} & \ = \ 0,\, \mb{\pt{}  x_{s}} \ = \ 0 &\text{if } s = \good \\
		\rv{x_s}  & \ = \ \rewFunction(s) + \sum_{s' \in S} \Prob(s, s') \cdot \rv{x_{s'}} &\text{for } s \in S \setminus \{\good\}\\
%		\mb{\pt{}  x_{\good}} &= 0 \\
%		\mb{\pt{}  x_s} &= \pt{} rew(s) + \sum_{s' \in S} \left(
%		\pt{}  \Prob(s, s') \cdot \rv{x_{s'}} +
%		\Prob(s, s') \cdot \mb{\pt{}  x_{s'}} \right) \; \;
%		&\text{for } s \in S \setminus \{\good\}
		\mb{\pt{}  x_s} & \ = \ \sum_{s' \in S} \big(
		\pt{}  \Prob(s, s') \cdot \rv{x_{s'}} +
		\Prob(s, s') \cdot \mb{\pt{}  x_{s'}} \big) \; \;
		&\text{for } s \in S \setminus \{\good\}.
	\end{align*}
%	where \(\pt{}  \Prob(s, s')\) and \(\pt rew(s)\) are the derivatives w.r.t.\ \(p\) of respectively \(\Prob(s, s')\) and \(rew(s)\) .
where \(\pt{}  \Prob(s, s')\) is the derivative of the probability function \(\Prob(s, s')\) w.r.t.\ \(p\).
\end{restatable}
\noindent
Note that we obtain the derivative for $\rv{x_s}$, \ie $\mb{\pt{} x_s}$, by applying the sum rule and the product rule to $\rv{x_s}$.
This equation system is equivalent to an equation system for POMDPs in~\cite[p.47-48]{Aberdeen2003}.
%\jm{Obtained from what? I guess from the system in \cite{Aberdeen2003} and the other references.} \js{I changed the sentence}
%It is is equivalent to the equation system in e.g.~\cite{Aberdeen2003}.
We remark that the equation system is linear with coefficients in a polynomial ring.
However, if the parameters are considered to be variables, then the system of equations is nonlinear (and nonconvex)~\cite{DBLP:conf/tacas/Cubuktepe0JKPPT17}.
Observe that the equations for $\rv{x_s}$ do not depend on the equations for $\mb{\pt{}  x_s}$ and thus can be solved independently first.
The equations for $\rv{x_s}$ have a unique solution which coincides with $\solRew{s}{\good}$.
This is a known result for MCs~\cite[Ch.\ 10]{BK08} and carries over to pMCs~\cite{DBLP:phd/dnb/Junges20}.
We show below that the equation system for $\mb{\pt{}  x_s}$ has a unique solution as well and yields the derivative $\pt{}\solRew{s}{\good}$.

\begin{example}
For our running example we obtain the following equation system:
%%	The system of equations for the partial derivative consists of the system of equations:
%%	\begin{align*}
%%		\mb{\pt{} x_0} &= 1 \cdot \rv{x_1} + p \cdot \mb{\pt{} x_1}  + -1 \cdot \rv{x_2} + (1-p) \cdot \mb{\pt{} x_2} \\
%%		\mb{\pt{} x_1} &= 1 \cdot \rv{x_2} + p \cdot \mb{\pt{} x_2}  + -1 \cdot \rv{x_3} + (1-p) \cdot \mb{\pt{} x_3} \\
%%		\mb{\pt{} x_2} &= 1 \cdot \mb{\pt{} x_{\good}} \\
%%		\mb{\pt{} x_3} &= 1 \cdot \mb{\pt{} x_{\good}} \\
%%		\mb{\pt{} x_{\good}} &= 0.
%%	\end{align*}
\begin{align*}
\rv{x_0} &= 0+p\cdot \rv{x_1}+(1{-}p)\cdot \rv{x_2} &\
\mb{\pt{} x_0} &= 1 \cdot \rv{x_1} + p \cdot \mb{\pt{} x_1}  + -1 \cdot \rv{x_2} + (1{-}p) \cdot \mb{\pt{} x_2} \\
\rv{x_1} &= 1+p\cdot \rv{x_2}+(1{-}p)\cdot \rv{x_3} &\
		\mb{\pt{} x_1} &= 1 \cdot \rv{x_2} + p \cdot \mb{\pt{} x_2}  + -1 \cdot \rv{x_3} + (1{-}p) \cdot \mb{\pt{} x_3} \\
\rv{x_2} &= 2 + 1\cdot \rv{x_{\good}} &\
		\mb{\pt{} x_2} &= 1 \cdot \mb{\pt{} x_{\good}} \\
\rv{x_3} &= 3 + 1\cdot \rv{x_{\good}} &\
		\mb{\pt{} x_3} &= 1 \cdot \mb{\pt{} x_{\good}} \\
\rv{x_{\good}} &=0 &\ 		\mb{\pt{} x_{\good}} &= 0.
	\end{align*}
Solving these equations yields $\rv{x_0} = {-}p^2+2 \cdot p+2$, the expected reward function $\solRew{s_0}{\good}$, see~\cref{ex:pmc:er}, and $\mb{\pt{} x_0} = -2\cdot p+2$, \ie, $\pt{}\solRewWithM{s_0}{\good}{}$.
\end{example}

\begin{restatable}{theorem}{thmdsunique}
\label{thm:ds_unique}
The equation system of \cref{def:soed} has exactly one solution: $\rv{x_s}$ equals $\solRew{s}{\good}$ and $\mb{\pt{} x_s}$ equals $\pt{}\solRewWithM{s}{\good}{}$ for each $s \in S$.
%	where \(A\) is the constrained transition matrix for $\pdtmc$ in which the row and column for \(\good\) is omitted.
%This system of equations has exactly one solution.
\end{restatable}
\noindent
\iftoggle{TR}{% if technical report
The proof is given in \cref{prf:thm:ds_unique}.
}{% else
The proof is given in the extended version~\cite{??}.
\jm{Are we going to upload an extended version soon enough?}
}

From a computational point, we notice that computing $\pt{}\solRewWithM{s}{\good}{}$ by solving the equation system (over the field of rational functions $\R(V)$) is intractable, as this function may be exponential in the number of parameters.
Matters appear worse as we aim to compute the derivative w.r.t.\ to a subset of the parameters $V' \subseteq V$, rather than with respect to a single parameter.
However, we observe that, for a gradient descent, we are only interested in computing  $\left(\pt{}\solRewWithM{s}{\good}{}\right)[u]$, and the equation system can be solved efficiently when we substitute all $\Prob(s, s')$ by $\Prob(s, s')[u]$ and solve for $\left(\pt{}\solRewWithM{s}{\good}{}\right)[u]$ using constant coefficients from the rationals or reals\footnote{In our implementation, we support exact rationals or floating point arithmetic.}.
Furthermore, as the $\rv{x_s}$ variables can be solved independently of the $\mb{\pt{} x_s}$ variables, we first solve the $\rv{x_s}$-equation system with $|S|$ variables and equations. In a second step, we construct for every $p \in V'$ an equation system (with $|S|$ variables and equations) by directly substituting the $\rv{x_s}$ variables with the expected reward $\solRewWithM{s}{\good}{}[u]$.
In total, this means that we evaluate $(|V'| + 1)$ equation systems with $|S|$ equations and variables each.

\subsection{Derived Automaton}
We now show that an alternative way to obtain \(\pt{}\solRewWithM{s}{\good}{}\) is by the standard equation system for \solRewWithM{s}{\good}{} on the \enquote{derivative} of pMC \pdtmc.
To that end, we mildly generalize pMCs to (parametric) weighted automata~\cite{DKV09} and show that we can describe \enquote{taking the derivative} as an operation on these weighted automata.
We do so by relaxing our parametric probability distributions by dropping the requirement that $0 \leq \distFunc(\distDomElem)[u] \leq 1$; in particular, negative real values are allowed.
These functions are called quasi-distributions as $\sum_{\distDomElem\in\distDom}\distFunc(\distDomElem)[u]=1$ still holds.
Let $\wfaDistr(\distDom) \subset \Vect(X)$ denote the set of quasi-distributions.
%\jm{I don't really get why we keep a restriction on the WFA. Sure, the given construction results in such a function. But do we want it in the definition of WFA?}

\begin{definition} \label{def:wfa}
	A \emph{weighted finite automaton (WFA)} is a tuple $\wfaInit$ where $S$, $\sinit$, $T$, $\Paramvar$ are as in \cref{def:pmc} and $\wfaTrans \colon S \to \wfaDistr(S)$.
\end{definition}

\begin{example}
\Cref{fig:wfa} depicts WFA $\wfa$ with single parameter $p$.
Note that some of the transitions are labelled with $p$ and $1{-}p$ (as in \cref{fig:pmc}).
We will later explain the relation of this WFA to the pMC in \cref{fig:pmc}.
\end{example}
Instead of creating a system of equations to compute the derivative, we can alternatively construct an automaton which has the derivative as its semantics.
This is called the \emph{\dpmctxt}.
Intuitively, the automaton \(\dervar{p} \pdtmc\) of a pMC \(\pdtmc\) is constructed by applying product and sum rules directly to \(\MC\).%
% NOTE: I don't think this is very important and kinda obvious for the people who know about it. -Joshua
%\footnote{
%	Using the semiring of the polynomials with a derivation operation, the \dpmctxt can be
%	defined more generally.
%}

\begin{definition}\label{def:dpmc}
Let \(\pDtmcInit\) be a pMC with reward function \(\rewFunction\) and let \(p \in \Paramvar\) a parameter.
The \emph{\dpmctxt} of \pdtmc w.r.t. p is the WFA \(\dpmcInitPrime\) with the reward function \(\rewFunction'\) where
	\begin{compactitem}
		\item $S' = S \, \dot{\cup} \, \pt{}  S$ with $\pt{}  S = \{\, \pt{}  s \mid s \in S \, \}$,
		\item the transition function $\dpmcTrans$ is given by: \begin{align*}
		\dpmcTrans(s, t) \ = \ \begin{cases}
			\probdtmc(s, t) & \text{if } s, t \in S,\\
			\probdtmc(s',t') & \text{if } s, t \in \pt{}  S \text{ and } s=\pt{}s' \text{ and } t=\pt{}t',\\
			\pt{}  \probdtmc(s', t) & \text{if } s \in \dervar{p} S \text{ and } s = \pt{} s' \text{ and } t \in S,\\
			0 & \text{otherwise,}
%			0 & \text{if } s \in S \text{ and } t \in \dervar{p} S,
		\end{cases}
		\end{align*}
		\item the reward function $\rewFunction'$ is given by $\rewFunction'(s) = \rewFunction(s)$ for $s \in S$ and $\rewFunction'(s) = 0$ for $s \in \dervar{p} S$.
	\end{compactitem}
\end{definition}

The intuition behind this derived automaton is as follows.
\enquote{Deriving} the state \(\rv{s} \in S\) with respect to \(p \in V\) yields the new state \(\mbdervar{p}{s}\).
For every transition \(\Prob(\rv{s}, \rv{s'}) \neq 0\) for \(\rv{s}, \rv{s'}\in	S\), we \enquote{use the product rule} and add the transitions \(\Prob(\mbdervar{p}{s}, \mbdervar{p}{s'}) = \Prob(\rv{s}, \rv{s'}) \) and \(\Prob(\mbdervar{p}{s}, \rv{s'}) = \dervar{p} \Prob(\rv{s}, \rv{s'})\) to \(\dervar{p} \pdtmc\).

\begin{example}
Applying \cref{def:dpmc} to the pMC \pdtmc from \cref{fig:pmc} results in the
\dpmctxt \dpmc in \Cref{fig:dpmc}.
\end{example}

Note that although \dpmc is not a pMC as some transitions have negative weights, the parametric expected reward $\solRewWithM{\pt{} \sinit}{\good}{\dpmc}$ can be computed as in~\cref{def:soed} as we restrict ourselves to graph-preserving regions, ensuring continuously differentiability of $\solRewWithM{s}{\good}{\pdtmc}$.
The derivative of the expected reward in \pdtmc %, denoted as $\pt{}\solRewWithM{\sinit}{\good}{\pdtmc}$,
can now be obtained as the parametric expected reward ($\solRewWithM{\pt{} \sinit}{\good}{\dpmc}$) in \dpmc.

\begin{proposition}
For each pMC \pdtmc we have: $\solRewWithM{\pt{} \sinit}{\good}{\dpmc} \ = \  \pt{}\solRewWithM{ \sinit}{\good}{\pdtmc}.$
\end{proposition}
\noindent
Stated in words, the expected reward of the derived automaton $\pt{\pdtmc}$ equals the partial derivative of the expected reward of the pMC \pdtmc.

%%\sj{Why is it good to have this automaton?}
%%\lh{It's cool}
%%\sj{Ideas: Can we apply bisim on this? Is that ever a good idea? Questions for the future: What about PLA on this?}

\section{Gradient Descent}
\label{sec:gd}
Gradient descent (\gd) is a first-order\footnote{It is only based on the first derivative and not on higher ones.} optimization technique to maximize an objective function $f(u)$.
It updates the \gd parameters in the direction of its gradient $\pt f(u)$.
We want to use \gd to solve the problem introduced in \cref{subsec:problemstatement}, \ie, given $\lambda \geq 0$, and comparison operator $\sim$, find an instantiation $u \in R$ with: $\solRew{s_0}{\good}[u] \sim \lambda$.

We consider several \gd update methods (\plaingd, \momentum \gd~\cite{DBLP:books/lib/Rumelhart89}, and Nesterov accelerated \gd~\cite{nesterov1983method, DBLP:conf/icml/SutskeverMDH13}, \rmsprop~\cite{RMSProp}, \adam~\cite{DBLP:journals/corr/KingmaB14}, and \radam~\cite{DBLP:conf/iclr/LiuJHCLG020}).
Three variants of \gd are common in the literature.
\emph{Batch \gd} computes the gradient of $f$ w.r.t. all parameters.
In contrast, \emph{stochastic \gd} performs updates for each parameter separately.
\emph{Mini-batch \gd} sits in between and performs an update for a subset of parameters.
We describe the \gd update methods w.r.t.\ stochastic \gd, \ie, at step $t$ we update the instantiation at parameter $p_{i(t)}$, while the other valuations remain the same.
We update the parameters in round-robin fashion: $i(t) = t \mod |V|$.
Clearly, stochastic \gd can be extended to mini-batch and batch \gd, by updating more/all parameters at the same time.
We assume that the objective function $f$, starting instantiation $u$, and bound $\lambda$ are given and focus on ${\sim} = {>}$.
\Cref{alg:gd} shows the algorithm to find a feasible solution.
First of all, we discuss Plain \gd, after which we consider other existing \gd update methods.
Finally, we discuss several region restriction methods to deal with parameter regions.
%All \gd updateand region restriction methods originate from literature.

\begin{algorithm}[t]
	\caption{\gd}
	\label{alg:gd}
	\begin{algorithmic}[1]
		\WHILE {$f[u] \leq \lambda$}
		\IF {$u$ is a local optimum}
		\STATE  pick new $u$  \label{alg:line:randomu}
		\ENDIF
		\STATE update $u$ with GD-method \label{alg:line:gdmethod}
		\ENDWHILE
		\RETURN $u$
	\end{algorithmic}
\end{algorithm}

\subsection{\plaingd}
\emph{\plaingd} is the simplest type of \gd.
A fixed learning rate $\eta$ is used to determine the step size taken to reach a (local) maximum.
The parameter $p_i$ gets updated in $u$ based on $\pti f[u]$ as follows:
\[
	u^{t+1}_{i} = u^{t}_{i} + \eta \cdot \pti{f}[u^t_i] ,
\]
where $u^t_i = u^t(p_i)$, i.e., the value of $p_i$ with instantiation $u^t$.
\begin{figure}[t]
	\centering
\subfigure[Plain]{
\scalebox{0.3}{
	\begin{tikzpicture}
		\begin{axis}[axis lines=left, domain=0:3, restrict y to domain=0:6.5, xmin=0, xmax=3, ymin=0, ymax=6.5,xlabel=$p$, ylabel=$f(p)$, tick label style={font=\huge}, xlabel style={font= \huge, yshift=-10pt}, ylabel style={font=\huge}, xtick={0, 1, 2, 3}, minor xtick= {0.5, 1.5, 2.5}, minor ytick={1,3,5}]
		\addplot[mark=none, smooth] {0.5*x^4-4*x^3+9*x^2-4*x+2};
		\addplot[only marks, mark=halfcircle*,mark options={scale=2, fill=red},text mark as node=true,point meta=explicit symbolic,nodes near coords, every node near coord/.append style={yshift=0.1cm, xshift=-0.5cm}] coordinates {
			(1, 3.5) [\LARGE$t=0$]
			(1.4, 4.99)  [\LARGE$t=1$]
			(1.72, 5.7678) [\LARGE$t=2$] };
		\addplot[only marks, mark=halfcircle*,mark options={scale=2, fill=red},text mark as node=true,point meta=explicit symbolic,nodes near coords, every node near coord/.append style={yshift=0.2cm, xshift=0.5cm}] coordinates {
			(1.882177, 5.95845) [\LARGE $t=3$]};
		\addplot[only marks, mark=*,mark options={scale=2, fill=blue},text mark as node=true,point meta=explicit symbolic,nodes near coords, every node near coord/.append style={yshift=0.1cm, xshift=-0.5cm}] coordinates {					(2,6) 				};
		\end{axis}

	\end{tikzpicture}}
\label{fig:plain}
}~
\subfigure[\momentum]{
\label{fig:Momentum}
\scalebox{0.3}{
		\begin{tikzpicture}
			\begin{axis}[axis lines=left, domain=0:3, restrict y to domain=0:6.5, xmin=0, xmax=3, ymin=0, ymax=6.5,xlabel=$p$, ylabel=$f(p)$, tick label style={font=\huge}, xlabel style={font= \huge, yshift=-10pt}, ylabel style={font=\huge}, xtick={0, 1, 2, 3}, minor xtick= {0.5, 1.5, 2.5}, minor ytick={1,3,5}]
				\addplot[mark=none, smooth] {0.5*x^4-4*x^3+9*x^2-4*x+2};
				\addplot[only marks, mark=halfcircle*, mark options={scale=2, fill=red},text mark as node=true,point meta=explicit symbolic,nodes near coords, every node near coord/.append style={yshift=0.1cm, xshift=-0.5cm}] coordinates {
					(1, 3.5) [\LARGE $t=0$]
					(1.4, 4.99)  [\LARGE $t=1$]
					(2.0768, 5.98232) [\LARGE $t=2$]
					};
				\addplot[only marks, mark=*,mark options={scale=2, fill=blue},text mark as node=true,point meta=explicit symbolic,nodes near coords, every node near coord/.append style={yshift=0.1cm, xshift=-0.5cm}] coordinates {					(2,6) 				};

			\end{axis}
		\end{tikzpicture}}
}~
\subfigure[Nesterov accelerated]{
		\label{fig:nesterov}
		\scalebox{0.3}{
		\begin{tikzpicture}
			\begin{axis}[axis lines=left, domain=0:3, restrict y to domain=0:6.5, xmin=0, xmax=3, ymin=0, ymax=6.5,xlabel=$p$, ylabel=$f(p)$, tick label style={font=\huge}, xlabel style={font= \huge, yshift=-10pt}, ylabel style={font=\huge}, xtick={0, 1, 2, 3}, minor xtick= {0.5, 1.5, 2.5}, minor ytick={1,3,5}]
				\addplot[mark=none, smooth] {0.5*x^4-4*x^3+9*x^2-4*x+2};
				\addplot[only marks, mark=halfcircle*,mark options={scale=2, fill=red},text mark as node=true,point meta=explicit symbolic,nodes near coords, every node near coord/.append style={yshift=0.1cm, xshift=-0.5cm}] coordinates {
					(1, 3.5) [\LARGE $t=0$]
					(1.4, 4.99)  [\LARGE $t=1$]
					(1.901235, 5.97078) [\LARGE $t=2$] };
				\addplot[only marks, mark=*,mark options={scale=2, fill=blue},text mark as node=true,point meta=explicit symbolic,nodes near coords, every node near coord/.append style={yshift=0.1cm, xshift=-0.5cm}] coordinates {					(2,6) 				};
			\end{axis}
		\end{tikzpicture}
	}
}
%\subfigure[RMSProp]{
%\label{fig:rmsprop}
%\scalebox{0.5}{
%	\begin{tikzpicture}
%		\begin{axis}[axis lines=left, domain=0:3, restrict y to domain=0:6.5, xmin=0, xmax=3, ymin=0, ymax=6.5,xlabel=$p$, ylabel=$f(p)$, tick label style={font=\Large}, xlabel style={font= \Large, yshift=-10pt}, ylabel style={font=\Large}]
%			\addplot[mark=none, smooth] {0.5*x^4-4*x^3+9*x^2-4*x+2};
%			\addplot[only marks, mark=*,mark options={scale=2, fill=red},text mark as node=true,point meta=explicit symbolic,nodes near coords, every node near coord/.append style={yshift=0.1cm, xshift=-0.3cm}] coordinates {
%				(1, 3.5) [$t=0$]
%				(1.3162288, 4.70667)  [$t=1$]
%				(1.52940, 5.36013) [$t=2$]
%				(1.67891, 5.69602) [$t=3$]
%				(1.78458, 5.86186) [$t=4$]};
%		\end{axis}
%	\end{tikzpicture}
%}
%}
\caption{Different GD methods on $f$ for $R=[0,3]$}
\end{figure}
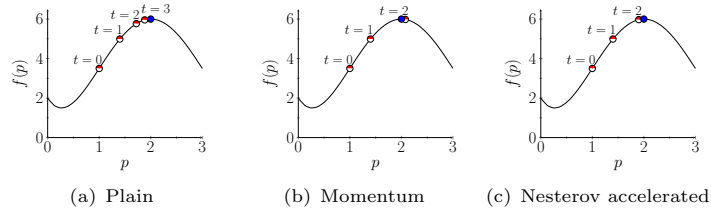

\begin{example}
	\label{ex:plain}
	Consider $f(p) = \frac{1}{2}p^4 - 4 p^3 + 9 p^2 - 4p + 2$ on a region $R = [0,3]$.
	Assume that our initial instantiation is $u^0(p) = 1$ and that we take $\eta = 0.1$ and $\lambda =
	5.9$.
	The red halfdots in~\cref{fig:plain} illustrate how the value of $p$ changes over time when using Plain \gd.
	The blue dot indicates the optimum.
	At $t=0$, the gradient is $4$ and so $p$ is updated to $1.4$.
	For $t=1$, the gradient is $3.17$, increasing $p$ again.
	This is repeated until at $t=3$, we have $f[u^t] = 5.96$.
	As this value exceeds $\lambda$, a feasible instantiation ($p=2.08$) is found.
%	gradient
%	t = 0: p = 1\\
%	t = 1: gradient = 4 mu = 0.1 -> 0.4 --> p = 1.4 f(p) = 4.99\\
%	t = 2: gradient = 3.168 mu = 0.1 -> 0.3168 --> p = 1.7168 f(p) = 5.7618\\
%	t = 3: gradient = 1.65377 mu =0.1 -> 0.165377 --> p=1.882177 f(p) = 5.95845
\end{example}

\subsection{GD Update Methods}
Intuitively, all \gd methods attempt to \enquote{guess} how the gradient will change by guiding the search for maxima based upon the past behaviour of the gradient.
Many \gd optimization methods exist and a recent overview is given by Ruder~\cite{ruder2016overview}.
We consider the following methods: \momentum, Nesterov accelerated \gd (\nag), \rmsprop, \adam, and \radam.
\momentum and \nag are classical and very similar to \plaingd.
The latter three are adaptive algorithms, \ie, their learning rate is changing over time and each parameter has its own learning rate.
Parameters with larger gradients have smaller learning rates than the ones with smaller gradients.
The latter three have been developed for machine learning purposes~\cite{DBLP:conf/iclr/LiuJHCLG020}.
We will elaborate on the \momentum and \nag method and briefly sketch the other methods.

\paragraph{\momentum~\cite{DBLP:books/lib/Rumelhart89}.}
Instead of only considering the current derivative, the \momentum method also takes into consideration previous derivatives.
They are weighted by the average decay factor $\gamma \in [0,1)$ (typically at least $0.9$).
This method uses an additional update vector $v$.
\momentum \gd adjusts the parameter value according to the following equation.
(Note that, if $\gamma = 0$, \momentum \gd is equal to \plaingd.)

\begin{align}
	v^{t+1}_{i} &\ = \ \gamma \cdot v^t_i + \eta \cdot \pti{f}[u^t_i] \\
	u^{t+1}_i & \ = \ u^t_i + v^{t+1}_i . \label{eq:Momentum:update}
\end{align}

\paragraph{Nesterov accelerated \gd (\nag)~\cite{nesterov1983method, DBLP:conf/icml/SutskeverMDH13}.}
As for \momentum \gd, \nag weighs the past steps by $\gamma$.
Additionally, it attempts to predict the future by guessing the next instantiation of $u$, denoted $u'$ (\cref{eq:nesterov:next}).
This should prevent us from moving to the other side of the optimum (\cref{ex:nesterov}).
As for \momentum, the instantiation is updated according to \cref{eq:Momentum:update}, whereas the update vector is obtained as in \cref{eq:nesterov}:
\begin{align}
	u'_j &\ = \ \begin{cases}
		u^t_j - \gamma \cdot v^t_j &\text{if } j=i\\
		u^t_j &\text{otherwise}
	\end{cases}\label{eq:nesterov:next}\\
	v^{t+1}_i &\ = \ \gamma \cdot v^t_i + \eta \cdot \pti{f}[u']. \label{eq:nesterov}
\end{align}

\begin{example}
	\label{ex:nesterov}
	Reconsider our running example.
	\Cref{fig:Momentum,fig:nesterov} show how the value of $p$ changes over time using \momentum \gd and \nag respectively.
	Note that for both methods we need one step less compared to \plaingd, i.e., a feasible instantiation is found at $t=2$.
	This is due to taking results of previous steps into account.
	Furthermore, observe that for \momentum \gd at $t=2$ the instantiation of $p$ actually passed the optimum, whereas for \nag this does not occur.
\end{example}

\paragraph{Adaptive methods.}
\emph{\rmsprop} (Root Mean Square Propagation)~\cite{RMSProp} is akin to \momentum and \nag, but its learning rate is adapted based on the previous squared gradient (\cref{eq:rmsprop:update}).
This squared gradient is recursively defined as the sum of $\beta \in [0, 1)$ times the past squared gradient, and $1-\beta$ times the current squared gradient (\cref{eq:rmsprop}).
$\beta$ is called the squared average decay.
In~\cref{eq:rmsprop:update}
a small amount $\epsilon > 0$ is added to the update vector at $p_i$ to avoid division by zero.

\begin{align}
	v^{t+1}_i &\ = \ \beta \cdot v^{t}_i + \left(1-\beta \right) \cdot \left(\pti{f[u]}\right)^2 \label{eq:rmsprop}\\
	u_i^{t+1} & \ = \ u^t_i + \dfrac{\mu}{\sqrt{v^t_i + \epsilon}} \cdot \pti{f[u]}. \label{eq:rmsprop:update}
\end{align}

In addition to the mean, \emph{\adam} (Adaptive Moment Estimation)~\cite{DBLP:journals/corr/KingmaB14} takes the second moment (the uncentered variance) of the gradients into account.
\emph{\radam} (Rectified \adam)~\cite{DBLP:conf/iclr/LiuJHCLG020} solves an issue with \adam in which the variance of learning rate is too large in the initial steps of the algorithm.

%In addition to the mean, \emph{\adam} (Adaptive Moment Estimation)~\cite{DBLP:journals/corr/KingmaB14} takes the second moment (the uncentered variance) of the gradients into account ($m_t$).
%
%\begin{align}
%	m^t_i &= \beta \cdot m^{t-1}_i + \left(1-\beta \right) \cdot \pti{f[u]} \label{eq:adam:m}\\
%	u_i &= u_i + \dfrac{\mu}{\sqrt{\hat{v}^t_i + \epsilon}} \cdot \hat{m}^t \label{eq:adam:update}
%\end{align}
%
%As \adam gets biased at both the initial time steps, and when $\beta$ and $\gamma$ are close to zero, $\hat{v}$ and $\hat{m}$ are computed as follows:
%\begin{align}
%	\hat{v}^t = \dfrac{v^t}{1-\gamma^t},\qquad  \hat{m}^t = \dfrac{m^t}{1-\beta^t},
%\end{align}
%where \(\gamma^t\ (\beta^t)\) denotes \(\gamma\ (\beta)\) to the power $t$.
%
%\emph{\radam} (Rectified \adam)~\cite{DBLP:conf/iclr/LiuJHCLG020} solves an issue with \adam in which the variance of learning rate is too large in the initial steps of the algorithm.

\paragraph{Sign methods~\cite{DBLP:journals/isci/MoulayLP19}.}
For the non-adaptive methods, we additionally implemented variants that only respect the signs of the gradients and not their magnitudes.
That is, we update the parameter as
\[
	u^{t+1}_{i} = u^{t}_{i} + \eta \cdot \text{sgn}(\pti{f}[u^t]) .
\]
Note that this implies we don't need to calculate the full gradient.

\subsection{Dealing with parameter regions}
So far we dealt with unconstrained \gd.
However, as a graph-preserving region $R$ is given, we need to deal with parameter values getting out of $R$. To do so, we discuss the following methods: Projection, Penalty Function, Barrier Function, and logistic Function.
Recall that, $R_i = [\text{\it lb}_i, \text{\it ub}_i]$ denotes the bound for parameter $p_i$ in region $R$.

\paragraph{Projection.}
The projection method acts as a hard wall around the region.
As soon as $u_i \not \in R_i$, $u_i$ gets set to the bound of the region, \ie,
${u_i^t}' = min(max(u_i^t, \text{\it lb}_i), \text{\it ub}_i)$.
Furthermore, if the parameter $p_i$ got out of the given region, we set its past gradients to 0, \ie $v_i^{t+1} = 0$.
%We refer to these methods as \projectionwith and \projectionwithout respectively.
\begin{example}
	\label{ex:projection}
	Reconsider our running example.
	However, now consider region $R' = [0.5,1.5]$.
	For $t=0$, the gradient is $4$, and $p$ is updated to $1.4$.
	For $t=1$, the gradient is $3.17$,
	yielding $p$ to be updated to $1.72$.
	As this is out of the region $R'$, $p$ is projected to $1.5$.
%	For \projectionwith, the gradient is projected to $0$.
\end{example}
\paragraph{Penalty function.}
The penalty method~\cite{smith1997penalty} transforms the constrained problem into an unconstrained one, by adding a penalty function to $f[u^t]$.
This penalty depends on how bad the violation is, e.g. what the difference is between $u_i$ and the bounds of $R_i$.
It can be interpreted as a red warning zone \emph{outside} of the region.
As this might yield non-graph-preserving instantiations, we do not further look into this.
%Let $\text{\it ub}_i$ $(\text{\it lb}_i)$ denote the upper (lower) bound for parameter $p_i$.
%\begin{align*}
%	\penfunc[t] &= \penmult \cdot \penvio[t] \\
%	dist = min(max())
%	\penvio[t] &= \begin{dcases}
%		0 & \text{ if } u^t \in R \\
%		? & otherwise
%	\end{dcases} \\
%\end{align*}
%
%The penalty function consists of the measure of violation ($\penvio$) and the penalty parameter ($\penmult$).

\paragraph{Barrier function.}
The barrier function~\cite{DBLP:books/daglib/0095265} (also called indicator function) works as a soft wall inside of the region, discouraging one to get to close to the wall.
It is independent of how bad the violation is.
%As the penalty function, the barrier function adds a term to $f[u^t]$.
%As with projection, the barrier function ensures our instantiation remains within the region.
We consider the log-barrier function for maximizing $f$ (see Eqs. (\ref{eq:barrier})-(\ref{eq:barrier:otherwise}))\footnote{When considering a minimization problem, $\barfunc$ is subtracted from $f$.}, as this yields a differentiable function.
The barrier function is weighted by $\mu \in [0,1]$.
The equations are:
\begin{align}
	f[u^t] &= f[u^t] + \mu \cdot \barfunc[u^t] \label{eq:barrier:f}\\
	\pti{f'}[u^t] &= \pti{f}[u^t] + \mu\cdot \pti{\barfunc}[u^t]\\
	\barfunc[u^t] &= \sum_i \barfunc_i[u^t]\label{eq:barrier}
\end{align}
%Moreover, we obtain \(f'[u^t] = f[u^t] + \mu \cdot \barfunc[u^t] \) and  $\pti{f'}[u^t] = \pti{f}[u^t] + \mu\cdot \pti{\barfunc^t}$, where $\barfunc[u^t]$ as defined in
% \ie, if $u_i^t \not \in R_i$, we set $\barfunc_i[u^t]$ to $-\infty$
%\begin{align}
%
%\end{align}
\begin{subnumcases}{\barfunc_i[u] = }
	\log(u_i - \text{\it lb}_i) & if $\text{\it lb}_i + \frac{\text{\it ub}_i - \text{\it lb}_i}{2} < u_i$ and $u_i \in R_i$ \label{eq:barrier:lb}
	\\
	\log(\text{\it ub}_i - u_i) & if $\text{\it lb}_i + \frac{\text{\it ub}_i - \text{\it lb}_i}{2} \geq u_i$ and $u_i \in R_i$ \label{eq:barrier:ub}
	\\
	-\infty & otherwise. \label{eq:barrier:otherwise}
\end{subnumcases}
\begin{subnumcases}{\pti{\barfunc_i}[u] = }
	\frac{1}{u_i - \text{\it lb}_i} & if $\text{\it lb}_i + \frac{\text{\it ub}_i - \text{\it lb}_i}{2} < u_i$ and $u_i \in R_i$ \label{eq:barrier:der:lb}
	\\
	\frac{1}{\text{\it ub}_i - u_i} & if $\text{\it lb}_i + \frac{\text{\it ub}_i - \text{\it lb}_i}{2} \geq u_i$ and $u_i \in R_i$ \label{eq:barrier:der:ub}
	\\
	\infty & otherwise. \label{eq:barrier:der:otherwise}
\end{subnumcases}

%Observe that when $\mu\rightarrow 0$, $\mu \cdot \barfunc^t \rightarrow -\infty$.
%Note that when the barrier function returns $-\infty$, we do not need to calculate the gradient.
Note that for higher learning rates, the barrier function might not be strong enough to prevent $u_i \not \in R$, see also the upcoming example.

\begin{example}
	\label{ex:barrier}
	Reconsider our running example with $\mu = 0.1$.
	We observe that at all $t$ where $u_i \in R_i$ case \cref{eq:barrier:ub} applies, so the barrier function is given by $\barfunc^t = \log(1.5-p)$.
	For learning rate $0.1$, at $t=0$, the gradient is $4 - \mu \cdot \frac{1}{1.5-p}$, so $p$ is updated to $1.38$.
	For $t=1$, the gradient is $0.24$. So $p$ is updated to $1.62$, which is outside region $R'$.
	When considering a smaller learning rate, \eg $0.01$, at $t=0$ $p$ is updated to $1.038$.
	This converges around $t=30$ with $p\approx1.46\in R'$.

\end{example}

\paragraph{Logistic function.}
For the logistic function, we map each restricted parameter $p_i$ to unrestricted parameter $q_i$ by using a sigmoid function~\cite{DBLP:conf/iwann/HanM95} (see~\cref{eq:sigmoid}) tailored to $R_i$.
We denote instantiations of $q$ with $u'$.
$u'_{i,0}$ is the value of the sigmoid's midpoint.
$u'$ gets updated according to the \gd method.
The gradient ($v'_i$) at $u'$ is computed according to~\cref{eq:sigmoid:gradient}. %Finally, the next instantiation of $u$ is obtained as in~\cref{eq:sigmoid:u},
\begin{align}
	u'_{i,0} &= \dfrac{\text{\it ub}_i - \text{\it lb}_i}{2} \nonumber \\
	u_i &= \frac{\text{\it ub}_i - \text{\it lb}_i}{1 + e ^{-(u'_i - u'_{i,0})}} + \text{\it lb}_i \label{eq:sigmoid}\\
	v'_i[u'] &= \frac{e ^{u'_i} \cdot v_i[u]}{(1+e^{u'_i})^2}.
	\label{eq:sigmoid:gradient}
\end{align}

\begin{example}
	\label{ex:logistic}
	Reconsider our running example.
	Let the learning rate be 0.1, and ${u}^{\prime 0}(q)=0.5$. The sigmoids midpoint is
	$u'_{i,0} = 0.5$.
%region is van 0.5-1.5
	For $t=0$, we have $u_i^0 = 1$.
	The gradient at this point ${v}_{\phantom{\prime}i}^{\prime 0}[{u}^{\prime 0}] = 0.94$, so $q$ is updated to $0.59$.
	Therefore, $p$ is set to $1.02$.
	At each iteration $p$ and $q$ get updated.
	E.g. at $t=100$, $q=3.63$ and $p=1.45$.

	\end{example}

\section{Empirical Evaluation}
\label{sec:empirical}
We implemented all gradient descent methods from \cref{sec:gd} in the probabilistic model checker \storm~\cite{DBLP:journals/corr/HenselJKQV20}.
All parameters, \ie batch-size, learning rate, average decay and squared average decay, are configurable via \storm's command line interface.
We evaluate the different gradient descent methods and compare them to two baselines:
One approach based on Quadratically-Constrained Quadratic Programming (\qcqp)~\cite{DBLP:conf/atva/CubuktepeJJKT18}, a convex optimization-based method. and the sampling-based approach Particle Swarm Optimization (\pso)~\cite{DBLP:conf/tase/ChenHHKQ013}.
These baselines are implemented in the tool \prophesy~\cite{DBLP:conf/cav/DehnertJJCVBKA15}.
All methods use the same version of \storm{} for model building, simplification, model checking, and solving of linear equation systems.
We specifically answer the following questions experimentally:
\begin{enumerate}
\item[Q1] Which region restriction method works best?
\item[Q2] Which \gd{} methods works best?
\item[Q3] How does \gd{} compare to previous techniques (\qcqp{} and \pso{})?
\end{enumerate}

%\footnote{Reviewers: The implementation is available at \url{https://github.com/moves-rwth/storm/pull/111/}. We will properly archive these changes via Zenodo.}.

%One sparse matrix is created per parameter and instantiated at the current position.
%Caching\sj{we cache the structure of the matrix as with pla?, I dont think it is clear here what we cache} is enabled to make subsequent calculations easier.
%\lh{Yes, I think so. gmm++ is doing the caching, so I'm not sure what exactly is cached. Otherwise we can just leave out the sentence.}
%\sj{Do you construct the equation system every time from scratch? :-)}

\subsection{Set-up}

We took the approach as described in \cref{subsec:equationsystem}, \ie, one sparse matrix is created per parameter and instantiated at the current position.
Our implementation works with Mini-Batch \gd as described above.
This means that we compute the derivative w.r.t. \(k\) parameters and then perform one step.
We allow for stochastic \gd and batch \gd by setting $k$ to $1$ or $|V|$, respectively.
%Also the region restriction method is set via the command line.

For the experiments, we solve equation systems with GMRES from the \texttt{gmm++} linear equation solver library included in \storm, which uses floating-point arithmetic.
All experiments run on a single thread and perform some preprocessing (\eg bisimulation minimization).
The times reported are the runtimes for \gd, \pso and \qcqp and do not include preprocessing.
We set a time-out of two hours.
We have used machines with an Intel Xeon Platinum 8160 CPU and 32GB of RAM.
In the comparisons with \qcqp and \pso, we report the average runtime over five runs.

%For simplicity, we have measured the entire user and system runtime of \storm and \prophesy, including model loading and bisimulation minimization.

\paragraph{Settings.}
For all constants except the learning rate, we chose the default from the literature (e.g.,
\cite{ruder2016overview, RMSProp, DBLP:journals/corr/KingmaB14,DBLP:conf/iclr/LiuJHCLG020}), \ie we set the batch size \(k\) to \(32\), average decay \(\gamma\) to \(0.9\) and squared average decay \(\beta\) to \(0.999\).
Whereas in the literature the learning rate is often set between \(0.001\) and \(0.1\), we stick to \(0.1\).
As we are interested in finding a feasible instantiation, we can take the risk of jumping over a local optimum due to a too high learning rate.
Also, our experiments show that lower learning rates slow down the search process (see~\cref{fig:restriction:all}).
Furthermore, we start at \(u_i = 0.5 + \varepsilon\) for all parameter $p_i$ with \(\varepsilon = 10^{-6}\), to overcome possible saddle points at \(p_i = 0.5\).
%This \(\varepsilon\) is introduced because MCs exist where \(p_i=0.5\) is a saddle point or a local optimum, leaving the initial gradient at zero, even if better points can be found elsewhere.
After every parameter has performed a step of less than \(10^{-6}\) in sequence, we conclude a local optimum has been found (we are aware this is an impatient criterion, tweaking this is a matter for further research).
When an infeasible local optimum is found, a new starting point is selected randomly (see \cref{alg:gd}, \cref{alg:line:randomu}).
Consequently, the \gd methods may yield different runtimes on different invocations on the same benchmark, though in practice we observe only a small deviation in the runtimes.
For the barrier region restriction method, we initially set $\mu$ to $0.1$. If no feasible solution is found, we divided $\mu$ by 10.
We continue this procedure until a feasible solution is found, or $\mu < 10^{-6}$.

\paragraph{Benchmarks.}
We consider pMCs obtained from POMDPs (cf.~\cite{DBLP:conf/uai/Junges0WQWK018}) and Bayesian networks (cf.~\cite{DBLP:conf/qest/SalmaniK20}) with a large number of parameters.
We took at least one variant of all POMDPs with reachability or expected reward properties from~ \cite{DBLP:journals/rts/Norman0Z17, DBLP:conf/atva/BorkJKQ20}, except for the dining cryptographer's protocol which has a constant reachability probability.
Furthermore, we took a medium and large Bayesian network from~\cite{bnrepository}.
We excluded the typical pMC examples~\cite{DBLP:conf/tacas/HartmannsKPQR19} with only two or four parameters.
We observed that for some benchmarks (e.g., \texttt{drone} and \texttt{refuel}) the optimum for some parameters is often at its bound.
We refer to these parameters as \easyparameters.%
\footnote{The feasibility problem remains a combinatorially hard problem, but the presence of easy parameters typically (but not always) indicates that the gradient remains (positive/negative) over the complete space.}

\begin{table}[t]
	\centering
	\caption{Model characteristics}
	\label{tab:model_info}
	{\footnotesize
	\scalebox{0.95}{
	\begin{tabular}{@{}l|lll|rrrr@{}}
		%		\toprule
		&Model			& Bound				& Instance  & States& Trans. 	& \(|V|\)	&\(|V_{\text{easy}}|\) \\
		\midrule
		\multirow{12}{*}{\rotatebox{90}{Reachability Probabilites}}
		&hailfinder		& \(\geq 0.145\)  	& (2000)  	& 1540 	& 324982	& 1249 		& 0 \\
		\cline{2-8}
		&nrp &  \(\leq0.001\) 	& (16,2)	& 787 	& 1602		& 95		& 32 \\
		&&  					& (16,5)	& 5806 	& 11685		& 704		& 340\\
		\cline{2-8}
		&drone
		& \(\geq0.85\)		& (5,1)		& 3678	& 27376		& 756		& 667 \\
		&&					& (5,2)		& 3678	& 27376		& 2640		& 404 \\
		\cline{2-8}
		&\multirow{2}{*}{\shortstack[l]{4x4grid\\\quad-avoid}}
		& \(\geq0.9\)		& (5)		& 1216	& 2495		& 99		& 42\\
		&&					& (10)		& 4931	& 9990		& 399		& 158\\
		\cline{2-8}
		&newgrid			& \(\geq0.99\)  	& (8,10)	& 30191	& 60410		& 399		& 244\\
		&&					& (15,10)	& 98441 & 196910	& 399		& 79\\
		\cline{2-8}
		&child			& \(\leq0.43\)		& (240)		& 243	& 3277		& 223		& 170 \\
		\cline{2-8}
		&refuel			& \(\geq0.35\) 		& (5,3)		& 1564	& 4206		& 452		& 317 \\
		&&					& (8,3)		& 7507	& 21468		& 794		& 570 \\\hline\hline
		\multirow{7}{*}{\rotatebox{90}{Expected Reward}}
		&\multirow{2}{*}{\shortstack[l]{network2\\\quad -prios}}
		&\(\leq0.1\) 		& (8,5, ps) & 397   & 2837		& 140 		& 128 \\
		&& \(\leq3.5\) 		& (8,5, dp) & 2822  & 69688		& 888		& 537 \\
		\cline{2-8}
		&samplerocks
		& 	\(\leq40\)					& (8) 		& 11278 & 25205		& 2844 		& 644 \\
				\cline{2-8}
		&	4x4grid		& \(\leq4.2\)		& (5)		& 1410 	& 2879		& 99		& 38 \\
		&&					& (10)		& 5780	& 11659		& 399		& 177 \\
				\cline{2-8}
		&maze2
		& \(\leq6\)			& (15)		& 5340	& 10799		& 2624		& 1257\\
		&&					& (50)		& 61000	& 121799	& 29749 	& N/A\\
		%		\bottomrule
	\end{tabular}
	}}
\end{table}

\Cref{tab:model_info} shows the benchmarks.
The first seven benchmarks consider reachability properties, whereas the latter four consider expected rewards.
The table includes the required property (Bound) and the instance of the benchmark.
For \texttt{network2-prios}, ``ps'' refers to successfully delivered packets and ``dp'' refers to dropped packets.
For each benchmark we denote the number of states, transitions and parameters after minimization, as well as the number of \easyparameters.
The entry N/A for \easyparameters means that all runs for \gd timed out, therefore, no feasible
instantiation was found and the number of \easyparameters could not be determined.

%\newpage
%\input{table2.tex}
%\input{table3.tex}

%We obtain most of our pMCs from POMDPs cf. \cite{DBLP:conf/uai/Junges0WQWK018} using the POMDP-to-pMC converter

%We took at least one variant of all POMDPs from the supporting material\footnote{\url{http://www.prismmodelchecker.org/files/rts-poptas/}} of \cite{DBLP:journals/rts/Norman0Z17} except for the Dining Cryptographer's Protocol because of its constant reachability probability.
%Additionally, we took at least one variant of every benchmark in the supporting material\footnote{\url{https://github.com/moves-rwth/indefinite-horizon-pomdps}} of \cite{DBLP:conf/atva/BorkJKQ20} that is not also in the previous set of benchmarks.
%
%
%In particular, of the Robot in a grid and the Maze POMDPs, we selected the larger variants, and of the Wireless Network Scheduling POMDP, we selected the smaller variant.

% obtained from \bnrepository~\cite{bnrepository}.

To obtain bounds for the feasible instantiations, we considered values close to known optima from the literature.
For those benchmarks where the optimal was not available, we approximated it by applying \gd several times and picking the optimum solution found.
We checked feasibility against the optimum-bounds, and the relaxed bounds, where we relaxed all bounds by 10\% and 20\%, respectively.
The plots for 10\% are similar to those for 20\% and therefore omitted.
%To evaluate solution quality versus computation time and mitigate effects caused by bounds close to the optimum, we also ran all experiments again where we relaxed all bounds by 10\% and 20\% respectively.

%We compare the performance of \texttt{\momentumsign} with \qcqp and \pso because it yields the strongest results on our set of benchmarks \lh{compare?}.

\subsection{Results}
\label{subsec:results}

Our experiments show that \gd can be used to find feasible parameter instantiations. In the following, we provide the numerical results and then answer the questions Q1--Q3 in the next paragraphs.

\paragraph{Numerical results.}
The scatter plots in \cref{fig:restriction:all} show how the different region restriction methods compare for \momentumsign and \adam.
Point \((x, y)\) denotes that the restriction method \bestrestrictionmethod took \(x\) seconds and the alternative took \(y\) seconds to find a feasible instantiation for the given \gd method.
The scatter plots in \cref{fig:othergd,fig:qcqppso} show how the different \gd methods and the baseline methods \qcqp and \pso (y-axis) compare to \bestmethod (x-axis), respectively.
Note that all scatter plots are log-log scale plots. %\footnote{During review, raw data is available at \url{https://github.com/moves-rwth/gradient-descent-experiments} and will be archived on Zenodo afterwards.}.
Point \((x, y)\) denotes that \bestmethod took \(x\) seconds and the alternative took \(y\) seconds to find a feasible instantiation.
All implicit vertical lines denote the same benchmark.
Points on the TO/MO line denote that the method has timed out or used too much memory and the ERR line denotes that the method has encountered some internal error.
The dashed lines denote differences of a factor 10 and 100.

%\subsection{Main findings}
%rom the results, we also draw conclusions regarding Q1--Q3 posed above. \begin{enumerate}
%\item[@Q1]
%From the discussed restriction methods, \emph{\bestrestrictionmethod} outperforms the barrier function.
%Furthermore, \bestrestrictionmethod performs slightly better than logistic function for \momentumsign.
%\item[@Q2]
%From the discussed \gd methods, \emph{\bestmethod} outperforms the more sophisticated ones.
%\item[@Q3]
%Finally, our experiments show that \emph{\gd is a competitive approach} to the state-of-the-art feasiblity method \qcqp and \pso.
%\end{enumerate}
%

\paragraph{Comparison of region restriction methods.}
\label{subsec:comp:restriction}
%\begin{table}
%	\centering
%	\begin{tabular}{r|l}
%		\textbf{GD method} & \textbf{best region restriction method}\\
%		\hline
%		\plaingd & projected with/without gradient\\
%		\plainsign& projected without gradient\\
%		\momentum & projected with gradient\\
%		\momentumsign & projected with gradient\\
%		Nesterov & projected with gradient\\
%		Nesterov-Sign & projected with gradient\\
%		\adam & projected with gradient \\
%		\radam & projected with gradient \\
%		\rmsprop & projected with gradient
%	\end{tabular}
%\label{tab:restriction}
%\caption{Best region restriction method for \gd methods with learning rate 0.1 }
%\end{table}
\begin{figure}[t]
	\centering
	\subfigure[\momentumsign]{
		\label{fig:restriction:momentum}
		\begin{minipage}{0.4\linewidth}
		% ADD FIGURE \bestrestrictionmethod VS all other methods on best non-adaptive method
			\includegraphics[width=\linewidth]{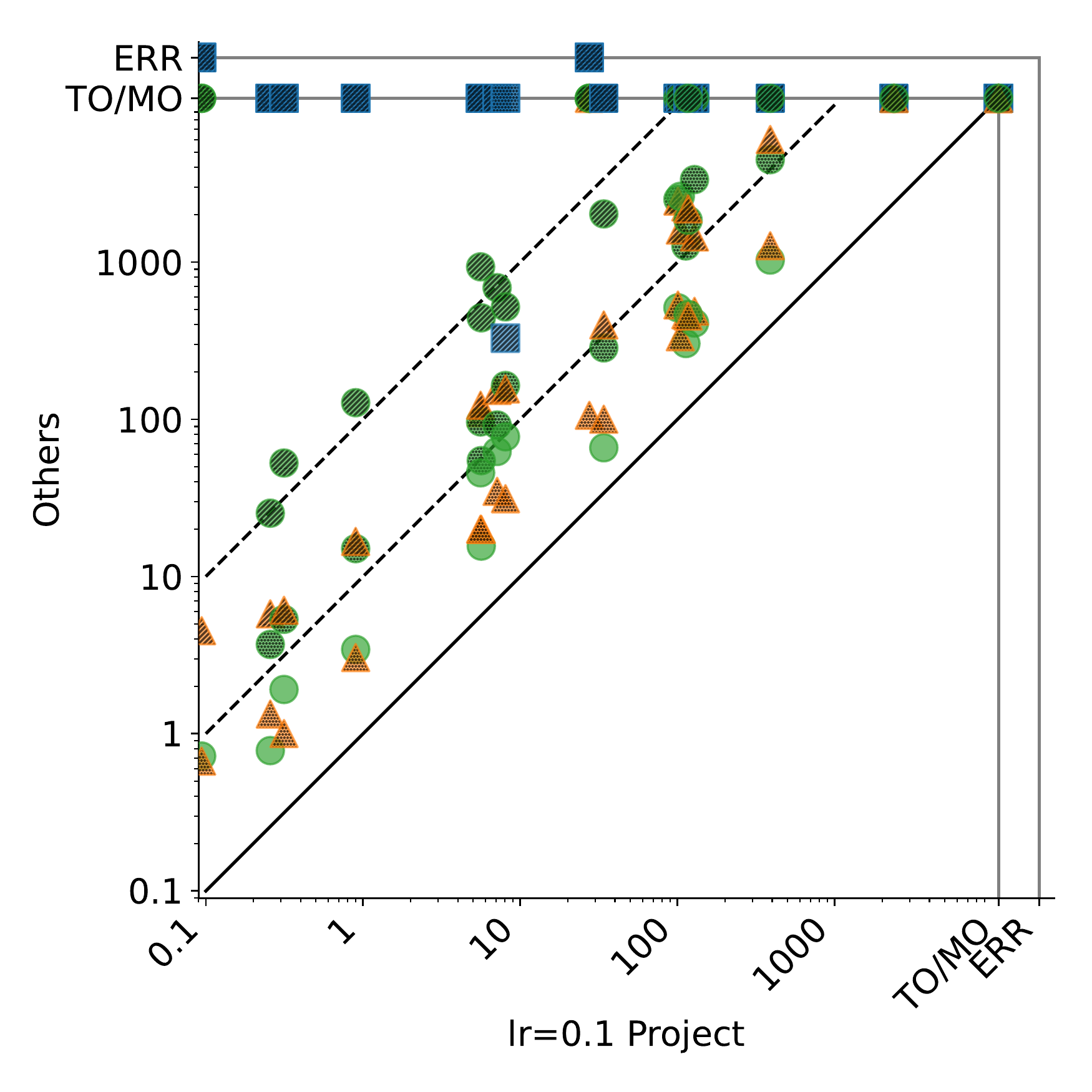}
		\end{minipage}
	}~
	\subfigure[\adam]{
		\begin{minipage}{0.4\linewidth}
				% ADD FIGURE \bestrestrictionmethod VS all other methods on \adam (or another adaptive method when this performs better)
			\includegraphics[width=\linewidth]{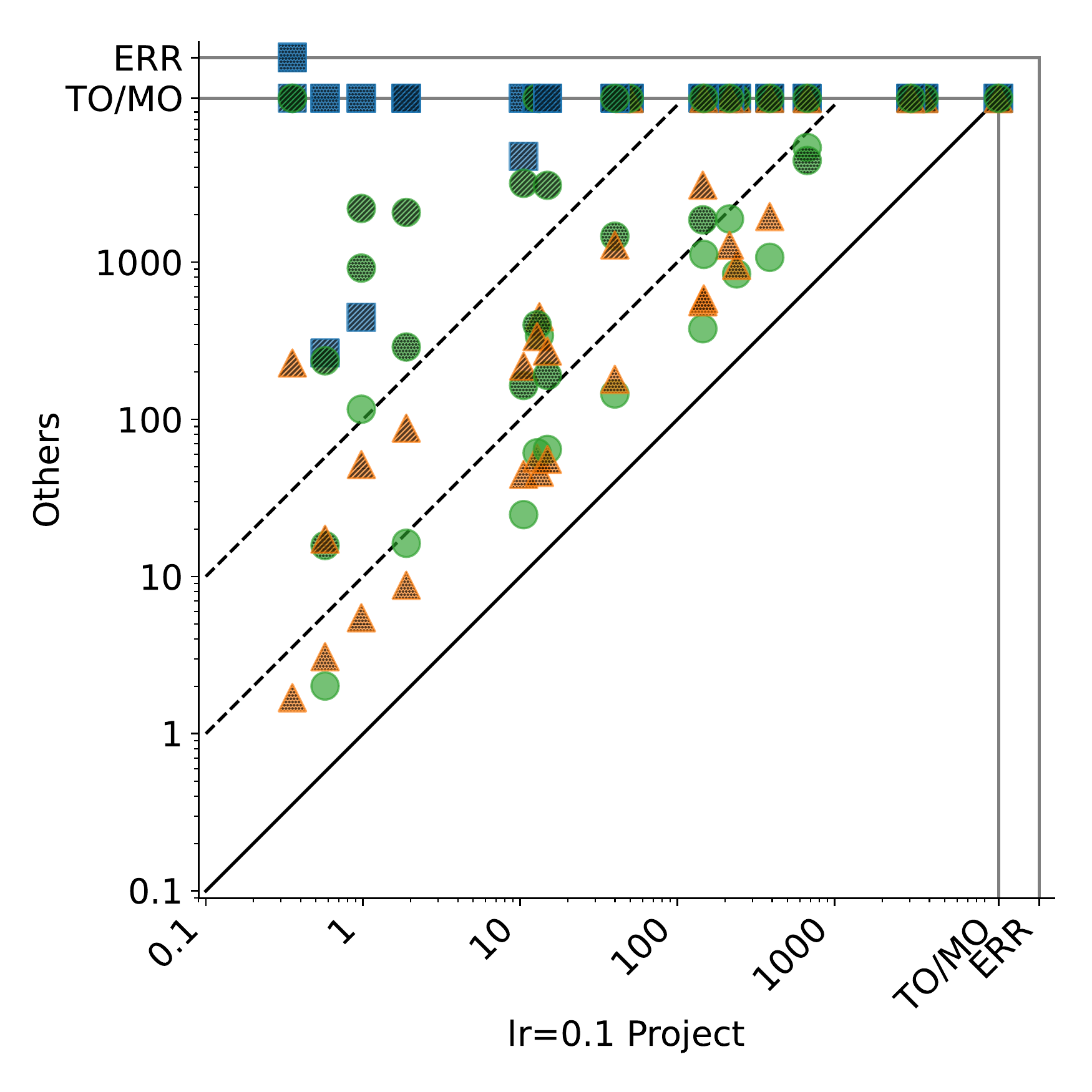}
		\end{minipage}
			\label{fig:restriction:adam}
	}\subfigure{
		\begin{minipage}{0.19\linewidth}
			\includegraphics[width=\linewidth]{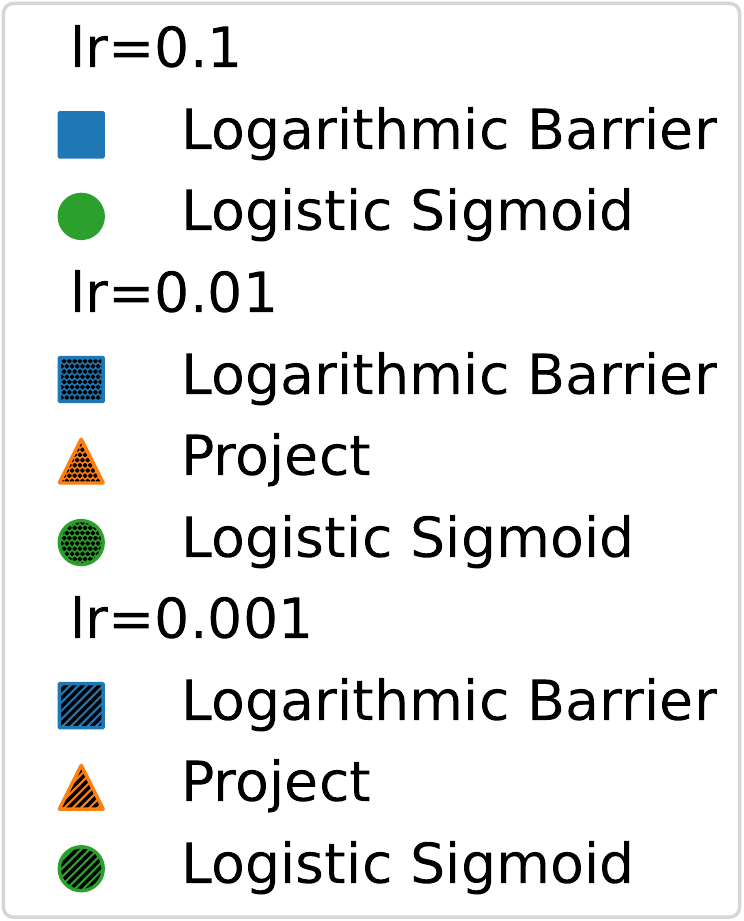}
		\end{minipage}

	}
	\caption{Comparison of different region restriction methods}
\label{fig:restriction:all}
\end{figure}

\Cref{fig:restriction:momentum} (\cref{fig:restriction:adam}) displays how \bestrestrictionmethod with learning rate 0.1 (x-axis) compares to all other restriction methods for the optimum-bounds of all benchmarks on \momentumsign (\adam).
The ERR line indicates that we found an infeasible parameter instantiation.
This occurs when the learning rate is too high, and thus the barrier function not strong enough (see also~\cref{ex:barrier}).
Imagine a vertical line through $x=0.1$.
This line represents the benchmark for which momentum-sign needed $\approx 0.1$ seconds.
We now obtain that the barrier function timed-out or threw an error for all learning rates.
%First of all, we observe that for learning rate 0.1 we obtain almost the same results for \projectionwithout and \projectionwith.

%The TO/MO cases are caused by randomly picking a new starting point, as an infeasible local optimum was found.
First of all, we observe that for \momentumsign the logistic-function is slightly outperformed by \projectionwith.
Secondly, we observe that for \adam the logistic-function is outperformed by \projectionwith often up to orders of magnitude.
Finally, we observe that for learning rate 0.1, the barrier function method is outperformed by \bestrestrictionmethod.
As many \easyparameters occur, the optima often lie at the edges of the region.
Therefore, we choose a relatively large learning rate.
The barrier function method tends to push us away from the edges, as the steps taken are too large, we cannot get close enough to the edge.

\paragraph{Comparison of \gd methods.}
\label{subsec:comp:gdmethods}
When comparing the different \gd Methods, we fix the region restriction method to \bestrestrictionmethod.
\begin{figure}[t]
	\centering
	\subfigure[All methods]{
		\label{fig:all}
			\begin{minipage}{0.4\linewidth}
			\includegraphics[width=\linewidth]{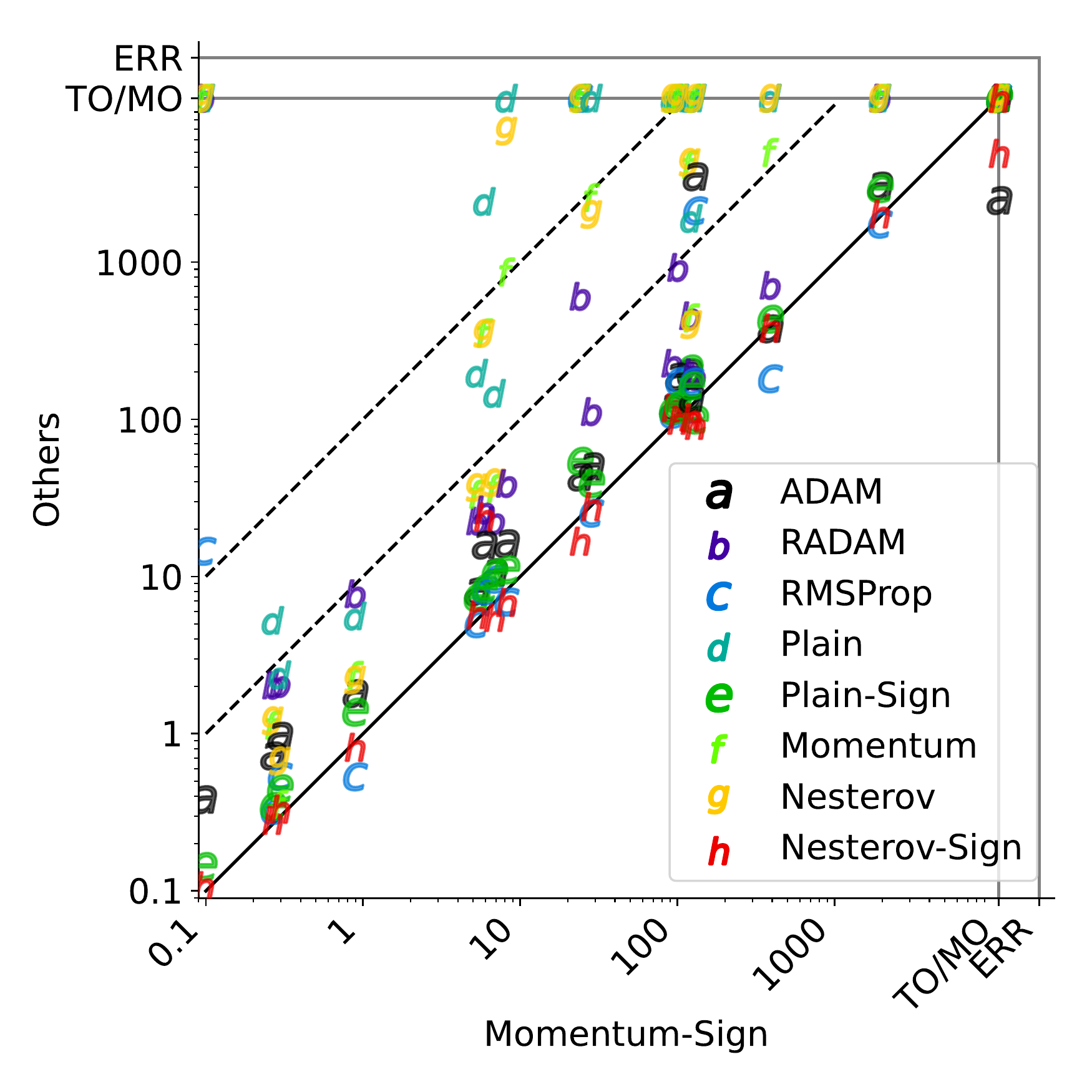}
		\end{minipage}
	}~
	\subfigure[\Bestmethod vs \momentum]{
		\begin{minipage}{0.4\linewidth}
	\includegraphics[width=\linewidth]{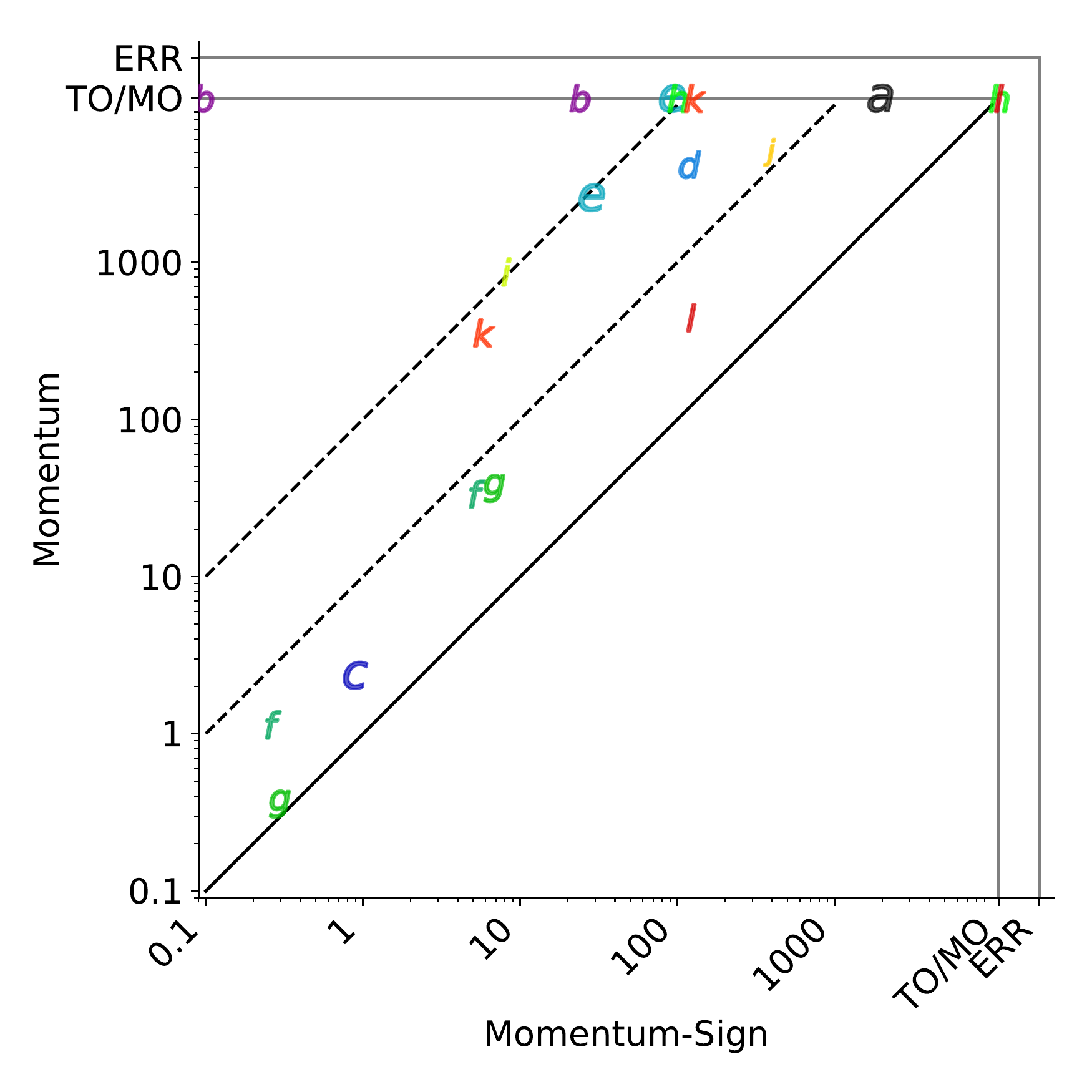}
		\end{minipage}
	\begin{minipage}{0.15\linewidth}
		\includegraphics[width=\linewidth]{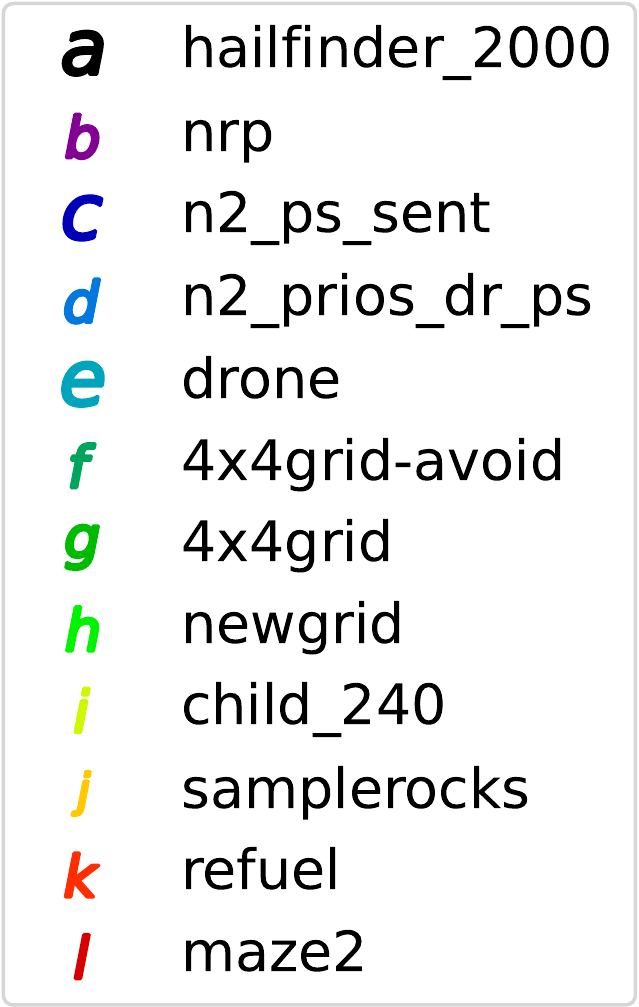}
	\end{minipage}

	\label{fig:best:vs:momentum}
	\label{fig:momentum}
	}
	\caption{Comparison of different \gd methods}
	\label{fig:othergd}
\end{figure}
\Cref{fig:all} displays how \bestmethod (x-axis) compares to all other methods for the optimum-bounds of all benchmarks.
First of all, we observe that \bestmethod typically obtains better runtimes compared to the adaptive methods (\rmsprop, \adam, \radam).
As our parameters occur with almost the same frequency, the adaptive methods are less suited for our benchmarks.
Secondly, we observe that for the non-adaptive methods, the methods where only the sign of the gradient is respected (and not the value gradient itself) often outperform their alternative.
This is caused by 1) the occurrence of the \easyparameters and 2) the influence a single parameter may have on the reachability probability/expected reward.
If a more influential parameter gets changed at the first parameter batch, this might yield a feasible solution before we have even updated all parameters.
Monotonicity could be a cause, and the  ordering of parameters on influentiallity needs further investigation (see~\cref{sec:conclusion}).
%However, for the benchmarks where most parameters are not \easyparameters (e.g. \texttt{samplerocks}), the colloquial version might perform better as can be seen in~\cref{fig:best:vs:momentum}.
%\js{Main conclusion: use non-adaptive methods, when user knows many parameters are \easyparameters, use sign-method, when many not \easyparameters, use colloquial}
%
%Nesterov \gd and \bestmethod are comparable, here the fact that the optimum often happens at the bound turns out to be of great influence.
%As this causes us to not be able to jump to the \enquote{other} side of the optimum.

\paragraph{Comparison to state-of-the-art feasibility methods.}
\label{subsec:comp:state-of-the-art}
\begin{figure}[t]
	\centering
	\label{fig:optimum}
	\begin{minipage}{0.8\linewidth}
	\includegraphics[width=0.5\linewidth]{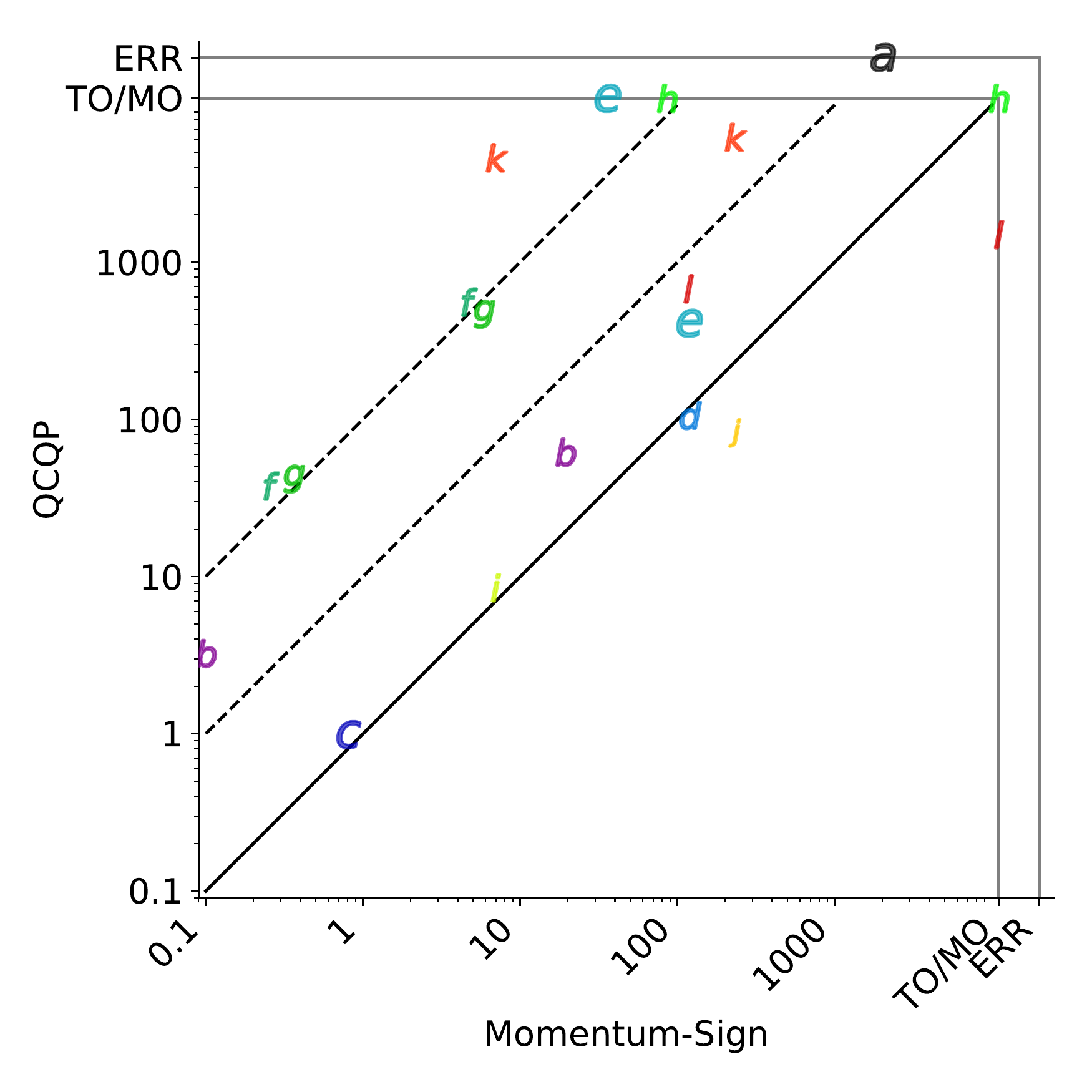}
	\includegraphics[width=0.5\linewidth]{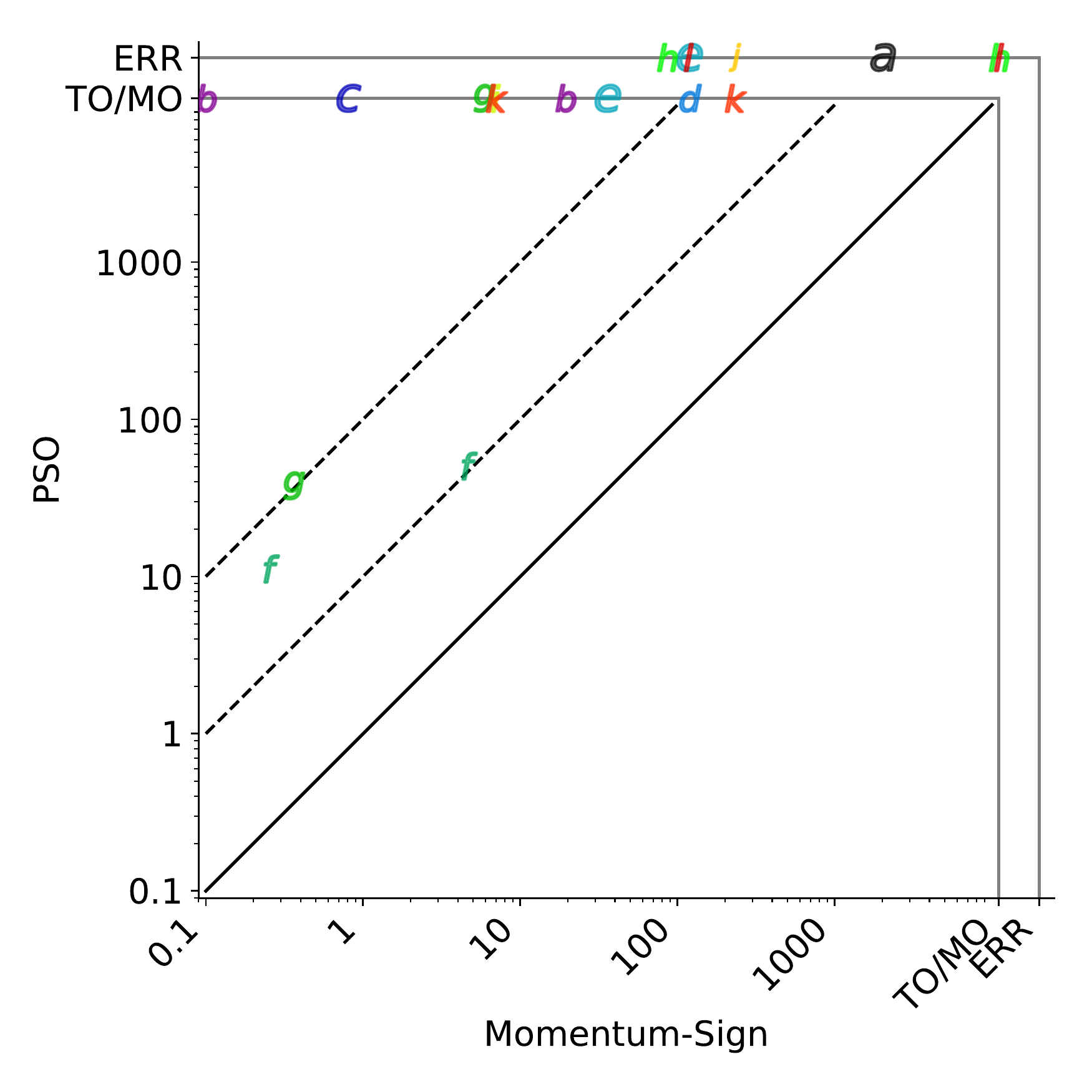}

	\includegraphics[width=0.5\linewidth]{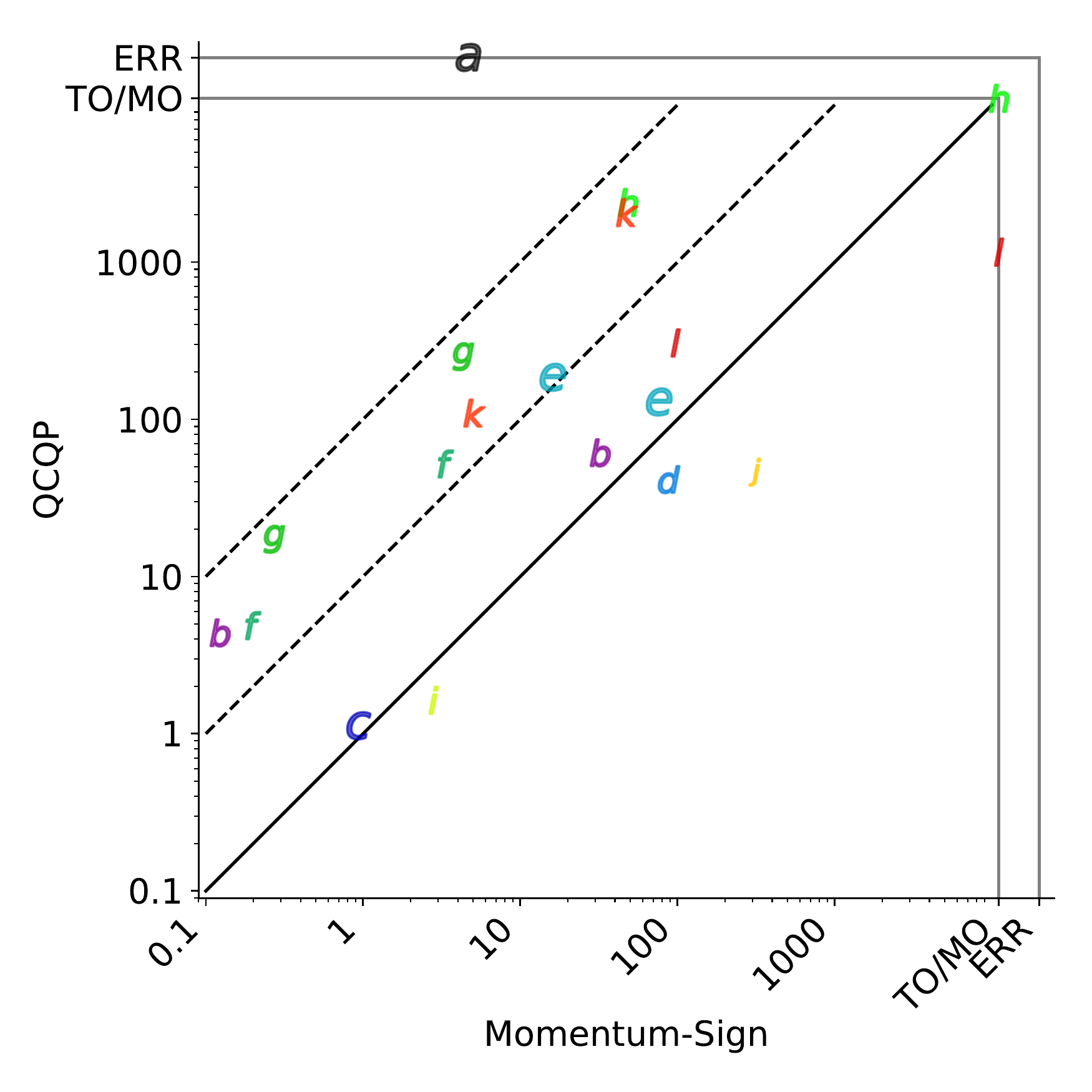}
	\includegraphics[width=0.5\linewidth]{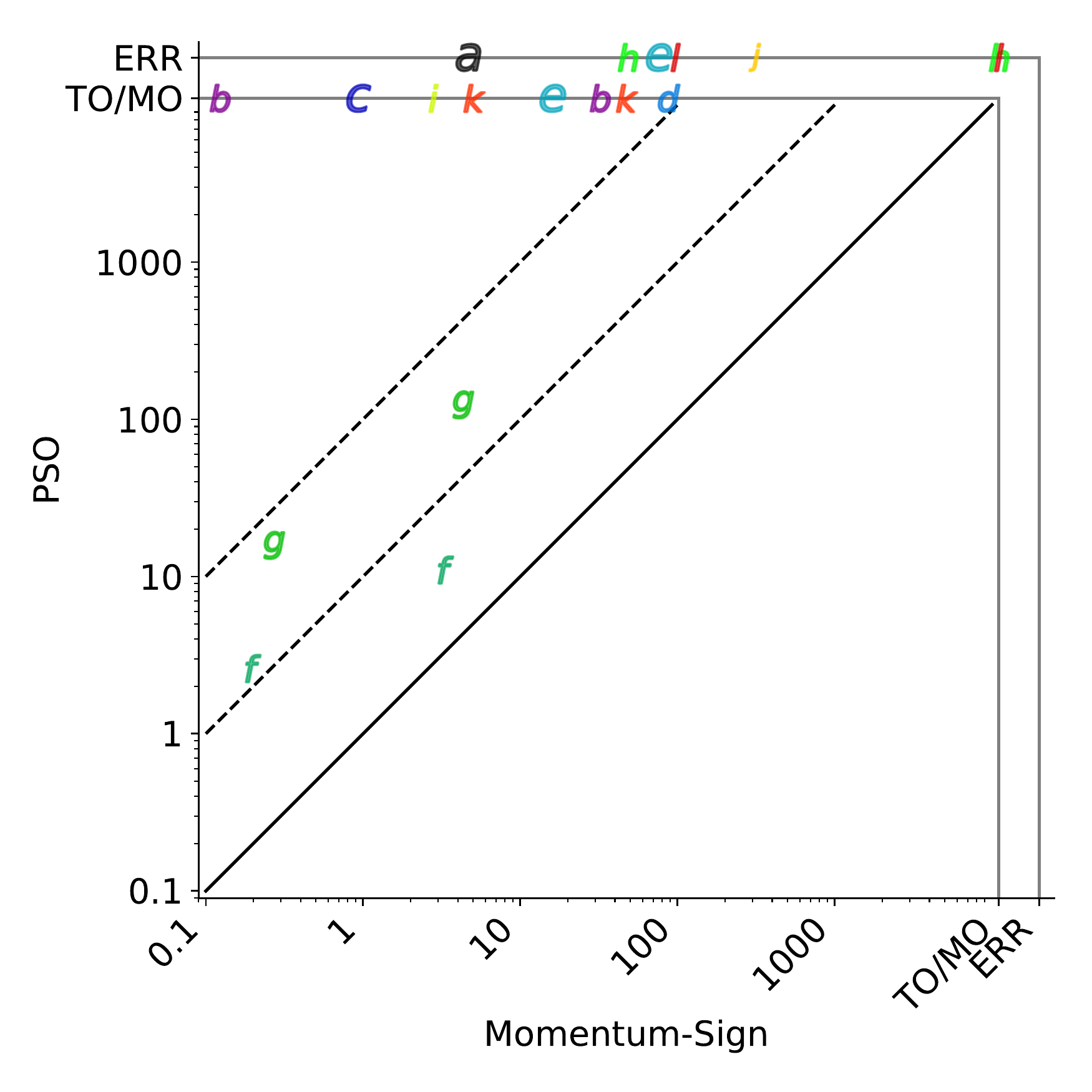}
\end{minipage}
\begin{minipage}{0.15\linewidth}
	\includegraphics[width=\linewidth]{scatter/legendmethods.pdf}
\end{minipage}

	\caption{Comparison of GD with \qcqp and \pso against optimum-bounds (upper) and 20\% relaxed-bounds (lower)}
	\label{fig:qcqppso}
\end{figure}
%In the scatter plots of \cref{fig:resqcqp, fig:respso} we compare the runtimes of \bestmethod to \qcqp and \pso respectively.
\Cref{fig:qcqppso} shows \bestmethod with \bestrestrictionmethod versus \qcqp and \pso respectively, on both the optimum-bounds (upper) and 20\% relaxed-bounds (lower).
%
%\Cref{fig:twenty} similarly show \bestmethod versus \qcqp on 20\% relaxed bounds.
%
%\begin{figure}[t]
%%	\includegraphics[width=0.49\linewidth]{scatter/momentum-qcqp-tenpercent.pdf}
%	\includegraphics[width=0.49\linewidth]{scatter/momentum-qcqp-twentypercent.pdf}
%	\caption{Comparison of GD with \qcqp on 10\% (left) and 20\% relaxed bounds}
%	\label{fig:twenty}
%
%\end{figure}
First of all, our experiments reveal that \bestmethod always outperforms \pso, on both the optimum-bounds and the relaxed-bounds.
Secondly, note that \pso throws an error during preprocessing of the MC on some benchmarks as they violate an implicit assumption by the PSO implementation.
Thirdly, \bestmethod outperforms \qcqp often by at least one order of magnitude.
Finally, we observe that \qcqp outperforms \bestmethod for the \texttt{samplerocks} benchmarks.
Based on the structure of the \texttt{samplerocks} benchmark, preprocessing with e.g. monotonicity checking might improve \bestmethod (see~\cref{sec:conclusion}).

\section{Related Work}
\label{sec:related}

\paragraph{Finding satisfying instantiations of parametric MCs.}
Parametric MCs~\cite{Daws04,DBLP:journals/fac/LanotteMT07} have received quite some attention.
The classical focus has been on computing closed forms for solution functions that map parameter values to expected rewards~\cite{Daws04,DBLP:conf/spin/HahnHZ09,DBLP:conf/icse/FilieriGT11,DBLP:conf/qest/JansenCVWAKB14,DBLP:journals/iandc/BaierHHJKK20,DBLP:journals/corr/abs-1903-07993, DBLP:conf/icse/FangCGA21}.
Feasibility as considered in this paper --- finding a satisfying instantiation --- and its extension to model repair \cite{DBLP:conf/tacas/BartocciGKRS11} has been formulated as a search problem before: 
Chen \emph{et al.}~\cite{DBLP:conf/tase/ChenHHKQ013} considered three different search methods: 
PSO, Markov Chain Monte Carlo and Cross-Entropy. 
In this context, PSO was most successful.
Model repair and feasibility have also been studied as optimization problems:~\cite{DBLP:conf/tacas/BartocciGKRS11} considered a one-shot encoding, whereas~\cite{DBLP:conf/tacas/Cubuktepe0JKPPT17,DBLP:conf/atva/CubuktepeJJKT18} took iterative approaches in which the encoding was simplified around a point to guide the search.  
Spel \emph{et al.}~\cite{DBLP:conf/atva/SpelJK19} present a graph-based heuristic to determine whether a pMC is monotonic, i.e., whether the gradient w.r.t.\ some parameter is positive on the complete parameter space. 
Chen \emph{et al.}~\cite{DBLP:conf/concur/ChenFRS14} analyze (non-controllable) perturbations in MCs from a robustness perspective.
Fast sampling of the parameter space and evaluating the corresponding pMCs is also a preprocessing step to other methods~\cite{DBLP:conf/nfm/HahnHZ11,DBLP:journals/corr/abs-1903-07993}. 
Storm offers optimized routines, and for large numbers of samples, just-in-time compilation is a feasible alternative~\cite{DBLP:conf/atva/GainerHS18}.

\noindent
\paragraph*{Controller synthesis under partial observability.}
The standard model for controller synthesis under partial observability are partially observable MDPs (POMDPs)~\cite{DBLP:journals/ai/KaelblingLC98}.
Controller synthesis in finite POMDPs can equivalently be reformulated as controller synthesis for infinitely large belief-MDPs.
Due to the curse of history, finding a feasible controller for a quantitative objective --- the setting discussed in this paper --- is undecidable~\cite{DBLP:journals/ai/MadaniHC03}.
%% , and in general we cannot hope to do better.
At the beginning of this millennium, this lead to trying to search for memoryless or small-memory controllers in POMDPs~\cite{DBLP:conf/uai/MeuleauKKC99}. 
Among others, the use of gradient descent methods to learn finite-state controllers for partially observable environments was explored by Meuleau \emph{et al.}~\cite{DBLP:conf/uai/MeuleauPKK99}.
This approach has further developed into deep learning for POMDPs, as e.g. used to learn Atari-games~\cite{DBLP:journals/corr/MnihKSGAWR13}. 
Some methods allow explicit extraction of the finite-state controllers~\cite{DBLP:conf/ijcai/CarrJT20}.
Those approaches are generally model-free --- they learn policies from sets of demonstrations or traces.
Closest to our approach is the work by Aberdeen~\cite{Aberdeen2003} in using a model-based approach to find memoryless strategies in POMDPs via gradient descent.
The major differences are in computing the gradients by using value-iteration and a softmax operation, and the use of stochastic gradient descent. 
The approach back then could and did not compare to the current state-of-the-art methods. 

Quickly afterwards, breakthroughs in point-based solvers~\cite{DBLP:conf/ijcai/PineauGT03,DBLP:journals/jair/SpaanV05} and Monte-Carlo methods for finding solutions~\cite{DBLP:conf/nips/SilverV10} shifted attention back to the belief-MDP~\cite{DBLP:conf/aaai/WalravenS17,DBLP:conf/ijcai/HorakBC18} (although some of those ideas also influenced the deep-RL community).
Likewise,  most recent support in the probabilistic model checkers PRISM~\cite{DBLP:journals/rts/Norman0Z17} and  Storm~\cite{DBLP:conf/atva/BorkJKQ20} is based on an abstraction of the belief-MDP~\cite{DBLP:journals/ior/Lovejoy91} and abstraction refinement.
The use of \cite{DBLP:journals/tac/WintererJWJTKB21} of game-based abstraction leads to non-randomized controllers. 
Winterer \emph{et al.}~\cite{DBLP:conf/nfm/Winterer00020} support a finite set of uniform randomizations.
In contrast, we consider an infinite combination of possibilities.
Likewise, Andriushenko \emph{et al.}~\cite{DBLP:conf/tacas/Andriushchenko021} recently consider syntax-guided synthesis for partial information controllers with a finite set of options.  

\section{Conclusion and Future Work}
\label{sec:conclusion}
This paper has shown that gradient descent often outperforms state-of-the-art methods for tackling the feasibility problem: find an the instance of a parametric Markov chain that satisfies a reachability objective.
As synthesizing a realizable controller with a fixed memory structure and a fixed set of potential actions can formally be described as feasibility synthesis in pMCs~\cite{DBLP:conf/uai/Junges0QWWK018}. Our approach supports the correct-by-construction synthesis of controllers for systems whose behavior is described by a stochastic process.
%This is equivalent~\cite{DBLP:conf/uai/Junges0WQWK018} to synthesizing a randomised finite-state controller in POMDPs.
Experiments showed that 1) projection outperforms other region restriction methods, 2) basic gradient descent methods perform better on our problem than more sophisticated ones, and 3) \bestmethod often outperforms \qcqp and \pso.
\paragraph{Outlook.}
As observed in the evaluation of the results, future work consists of extending the preprocessing of the parametric Markov chains with monotonicity checking and investigating a possible ordering of parameters based on the influence on the property.
Also, models with a large state space could be handled by e.g. using value iteration to solve the system of equations.
%Also one could look at ways to parallelize the computation within the minibatches.
Furthermore, questions regarding the \dpmctxt can be asked, e.g. regarding the applicability of bisimulation minimisation or parameter lifting~\cite{DBLP:conf/atva/QuatmannD0JK16}.

\paragraph{Data availability.}
The tools used and data generated in our experimental evaluation are archived at DOI \href{https://doi.org/10.5281/zenodo.5568910}{10.5281/5568910}~\cite{PaperArtifact}.

\clearpage
\pagebreak
\bibliographystyle{splncs04}
\bibliography{literature}

\begin{thebibliography}{10}
\providecommand{\url}[1]{\texttt{#1}}
\providecommand{\urlprefix}{URL }
\providecommand{\doi}[1]{https://doi.org/#1}

\bibitem{Aberdeen2003}
Aberdeen, D.A.: Policy-Gradient Algorithms for Partially Observable {M}arkov
  Decision Processes. Ph.D. thesis, The Australian National University (2003)

\bibitem{sygus}
Alur, R., Bod{\'{\i}}k, R., Dallal, E., Fisman, D., Garg, P., Juniwal, G.,
  Kress{-}Gazit, H., Madhusudan, P., Martin, M.M.K., Raghothaman, M., Saha, S.,
  Seshia, S.A., Singh, R., Solar{-}Lezama, A., Torlak, E., Udupa, A.:
  Syntax-guided synthesis. In: Dependable Software Systems Engineering, {NATO}
  Science for Peace and Security Series, {D:} Information and Communication
  Security, vol.~40, pp. 1--25. {IOS} Press (2015)

\bibitem{DBLP:conf/tacas/Andriushchenko021}
Andriushchenko, R., Ceska, M., Junges, S., Katoen, J.P.: Inductive synthesis
  for probabilistic programs reaches new horizons. In: {TACAS} {(1)}. {LNCS},
  vol. 12651, pp. 191--209. Springer (2021)

\bibitem{DBLP:journals/jacm/BaierGB12}
Baier, C., Gr{\"{o}}{\ss}er, M., Bertrand, N.: Probabilistic
  {\(\omega\)}-automata. J. {ACM}  \textbf{59}(1),  1:1--1:52 (2012)

\bibitem{DBLP:journals/iandc/BaierHHJKK20}
Baier, C., Hensel, C., Hutschenreiter, L., Junges, S., Katoen, J.P., Klein, J.:
  Parametric {M}arkov chains: {PCTL} complexity and fraction-free {G}aussian
  elimination. Inf. Comput.  \textbf{272},  104504 (2020)

\bibitem{BK08}
Baier, C., Katoen, J.P.: Principles of Model Checking. {MIT} Press (2008)

\bibitem{DBLP:conf/tacas/BartocciGKRS11}
Bartocci, E., Grosu, R., Katsaros, P., Ramakrishnan, C.R., Smolka, S.A.: Model
  repair for probabilistic systems. In: {TACAS}. {LNCS}, vol.~6605. Springer
  (2011)

\bibitem{DBLP:conf/atva/BorkJKQ20}
Bork, A., Junges, S., Katoen, J.P., Quatmann, T.: Verification of
  indefinite-horizon {POMDPs}. In: {ATVA}. {LNCS}, vol. 12302, pp. 288--304.
  Springer (2020)

\bibitem{DBLP:conf/ijcai/CarrJT20}
Carr, S., Jansen, N., Topcu, U.: Verifiable rnn-based policies for {POMDP}s
  under temporal logic constraints. In: {IJCAI}. pp. 4121--4127. ijcai.org
  (2020)

\bibitem{DBLP:conf/concur/ChenFRS14}
Chen, T., Feng, Y., Rosenblum, D.S., Su, G.: Perturbation analysis in
  verification of discrete-time {M}arkov chains. In: {CONCUR}. {LNCS},
  vol.~8704, pp. 218--233. Springer (2014)

\bibitem{DBLP:conf/tase/ChenHHKQ013}
Chen, T., Hahn, E.M., Han, T., Kwiatkowska, M.Z., Qu, H., Zhang, L.: Model
  repair for {M}arkov decision processes. In: {TASE}. {IEEE} (2013)

\bibitem{DBLP:conf/tacas/Cubuktepe0JKPPT17}
Cubuktepe, M., Jansen, N., Junges, S., Katoen, J.P., Papusha, I., Poonawala,
  H.A., Topcu, U.: Sequential convex programming for the efficient verification
  of parametric {MDP}s. In: TACAS. LNCS, vol. 10206, pp. 133--150 (2017)

\bibitem{DBLP:conf/atva/CubuktepeJJKT18}
Cubuktepe, M., Jansen, N., Junges, S., Katoen, J.P., Topcu, U.: Synthesis in
  p{MDP}s: {A} tale of 1001 parameters. In: {ATVA}. {LNCS}, vol. 11138, pp.
  160--176. Springer (2018)

\bibitem{Daws04}
Daws, C.: Symbolic and parametric model checking of discrete-time {M}arkov
  chains. In: ICTAC. LNCS, vol.~3407, pp. 280--294. Springer (2004)

\bibitem{DBLP:conf/cav/DehnertJJCVBKA15}
Dehnert, C., Junges, S., Jansen, N., Corzilius, F., Volk, M., Bruintjes, H.,
  Katoen, J.P., {\'{A}}brah{\'{a}}m, E.: Prophesy: {A} probabilistic parameter
  synthesis tool. In: {CAV} {(1)}. {LNCS}, vol.~9206. Springer (2015)

\bibitem{DKV09}
Droste, M., Kuich, W., Vogler, H.: Handbook of weighted automata. Springer
  (2009)

\bibitem{DBLP:conf/icse/FangCGA21}
Fang, X., Calinescu, R., Gerasimou, S., Alhwikem, F.: Fast parametric model
  checking through model fragmentation. In: {ICSE}. pp. 835--846. {IEEE} (2021)

\bibitem{DBLP:conf/icse/FilieriGT11}
Filieri, A., Ghezzi, C., Tamburrelli, G.: Run-time efficient probabilistic
  model checking. In: {ICSE}. {ACM} (2011)

\bibitem{DBLP:conf/cav/FremontS18}
Fremont, D.J., Seshia, S.A.: Reactive control improvisation. In: {CAV} {(1)}.
  {LNCS}, vol. 10981, pp. 307--326. Springer (2018)

\bibitem{DBLP:conf/atva/GainerHS18}
Gainer, P., Hahn, E.M., Schewe, S.: Accelerated model checking of parametric
  {M}arkov chains. In: {ATVA}. {LNCS}, vol. 11138. Springer (2018)

\bibitem{DBLP:conf/formats/GiroD07}
Giro, S., D'Argenio, P.R.: Quantitative model checking revisited: Neither
  decidable nor approximable. In: {FORMATS}. {LNCS}, vol.~4763, pp. 179--194.
  Springer (2007)

\bibitem{DBLP:conf/nfm/HahnHZ11}
Hahn, E.M., Han, T., Zhang, L.: Synthesis for {PCTL} in parametric {M}arkov
  decision processes. In: NFM. LNCS, vol.~6617, pp. 146--161. Springer (2011)

\bibitem{DBLP:conf/spin/HahnHZ09}
Hahn, E.M., Hermanns, H., Zhang, L.: Probabilistic reachability for parametric
  {M}arkov models. In: {SPIN}. {LNCS}, vol.~5578, pp. 88--106. Springer (2009)

\bibitem{DBLP:conf/iwann/HanM95}
Han, J., Moraga, C.: The influence of the sigmoid function parameters on the
  speed of backpropagation learning. In: {IWANN}. Lecture Notes in Computer
  Science, vol.~930, pp. 195--201. Springer (1995)

\bibitem{DBLP:conf/tacas/HartmannsKPQR19}
Hartmanns, A., Klauck, M., Parker, D., Quatmann, T., Ruijters, E.: The
  quantitative verification benchmark set. In: {TACAS}. LNCS, vol. 11427.
  Springer (2019)

\bibitem{PaperArtifact}
Heck, L., Spel, J., Junges, S., Moerman, J., Katoen, J.P.: {G}radient-{D}escent
  for {R}andomized {C}ontrollers under {P}artial {O}bservability ({A}rtifact).
  Zenodo (2021). \doi{10.4121/14910426}

\bibitem{DBLP:journals/corr/HenselJKQV20}
Hensel, C., Junges, S., Katoen, J.P., Quatmann, T., Volk, M.: The probabilistic
  model checker storm. CoRR  \textbf{abs/2002.07080} (2020)

\bibitem{DBLP:conf/ijcai/HorakBC18}
Hor{\'{a}}k, K., Bosansk{\'{y}}, B., Chatterjee, K.: {Goal-HSVI}: Heuristic
  search value iteration for goal {POMDP}s. In: {IJCAI}. pp. 4764--4770.
  ijcai.org (2018)

\bibitem{DBLP:conf/podc/IsraeliJ90}
Israeli, A., Jalfon, M.: Token management schemes and random walks yield
  self-stabilizing mutual exclusion. In: {PODC}. pp. 119--131. {ACM} (1990)

\bibitem{DBLP:conf/qest/JansenCVWAKB14}
Jansen, N., Corzilius, F., Volk, M., Wimmer, R., {\'{A}}brah{\'{a}}m, E.,
  Katoen, J.P., Becker, B.: Accelerating parametric probabilistic verification.
  In: {QEST}. {LNCS}, vol.~8657. Springer (2014)

\bibitem{DBLP:phd/dnb/Junges20}
Junges, S.: Parameter synthesis in {M}arkov models. Ph.D. thesis, {RWTH} Aachen
  University, Germany (2020)

\bibitem{DBLP:journals/corr/abs-1903-07993}
Junges, S., {\'{A}}brah{\'{a}}m, E., Hensel, C., Jansen, N., Katoen, J.P.,
  Quatmann, T., Volk, M.: Parameter synthesis for {M}arkov models. CoRR
  \textbf{abs/1903.07993} (2019)

\bibitem{DBLP:conf/uai/Junges0WQWK018}
Junges, S., Jansen, N., Wimmer, R., Quatmann, T., Winterer, L., Katoen, J.P.,
  Becker, B.: Finite-state controllers of {POMDP}s using parameter synthesis.
  In: {UAI}. {AUAI} Press (2018)

\bibitem{DBLP:journals/jcss/JungesK0W21}
Junges, S., Katoen, J.P., P{\'{e}}rez, G.A., Winkler, T.: The complexity of
  reachability in parametric {M}arkov decision processes. J. Comput. Syst. Sci.
   \textbf{119},  183--210 (2021)

\bibitem{DBLP:journals/ai/KaelblingLC98}
Kaelbling, L.P., Littman, M.L., Cassandra, A.R.: Planning and acting in
  partially observable stochastic domains. Artif. Intell.  \textbf{101}(1-2),
  99--134 (1998)

\bibitem{DBLP:journals/corr/KingmaB14}
Kingma, D.P., Ba, J.: Adam: {A} method for stochastic optimization. In: {ICLR}
  (Poster) (2015)

\bibitem{DBLP:conf/cav/KwiatkowskaNP11}
Kwiatkowska, M.Z., Norman, G., Parker, D.: {PRISM} 4.0: Verification of
  probabilistic real-time systems. In: {CAV}. {LNCS}, vol.~6806. Springer
  (2011)

\bibitem{DBLP:journals/fac/LanotteMT07}
Lanotte, R., Maggiolo{-}Schettini, A., Troina, A.: Parametric probabilistic
  transition systems for system design and analysis. Formal Aspects Comput.
  \textbf{19}(1),  93--109 (2007)

\bibitem{DBLP:conf/iclr/LiuJHCLG020}
Liu, L., Jiang, H., He, P., Chen, W., Liu, X., Gao, J., Han, J.: On the
  variance of the adaptive learning rate and beyond. In: {ICLR}. OpenReview.net
  (2020)

\bibitem{DBLP:journals/ior/Lovejoy91}
Lovejoy, W.S.: Computationally feasible bounds for partially observed {M}arkov
  decision processes. Oper. Res.  \textbf{39}(1),  162--175 (1991)

\bibitem{DBLP:journals/ai/MadaniHC03}
Madani, O., Hanks, S., Condon, A.: On the undecidability of probabilistic
  planning and related stochastic optimization problems. Artif. Intell.
  \textbf{147}(1-2),  5--34 (2003)

\bibitem{DBLP:conf/uai/MeuleauKKC99}
Meuleau, N., Kim, K., Kaelbling, L.P., Cassandra, A.R.: Solving {POMDP}s by
  searching the space of finite policies. In: {UAI}. pp. 417--426. Morgan
  Kaufmann (1999)

\bibitem{DBLP:conf/uai/MeuleauPKK99}
Meuleau, N., Peshkin, L., Kim, K., Kaelbling, L.P.: Learning finite-state
  controllers for partially observable environments. In: {UAI}. pp. 427--436.
  Morgan Kaufmann (1999)

\bibitem{DBLP:journals/corr/MnihKSGAWR13}
Mnih, V., Kavukcuoglu, K., Silver, D., Graves, A., Antonoglou, I., Wierstra,
  D., Riedmiller, M.A.: Playing {A}tari with deep reinforcement learning. CoRR
  \textbf{abs/1312.5602} (2013)

\bibitem{DBLP:journals/isci/MoulayLP19}
Moulay, E., L{\'{e}}chapp{\'{e}}, V., Plestan, F.: Properties of the sign
  gradient descent algorithms. Inf. Sci.  \textbf{492},  29--39 (2019)

\bibitem{nesterov1983method}
Nesterov, Y.E.: A method for solving the convex programming problem with
  convergence rate {$O(1/k\text{\textasciicircum}2)$}. In: Dokl. akad. nauk
  Sssr. vol.~269, pp. 543--547 (1983)

\bibitem{DBLP:journals/rts/Norman0Z17}
Norman, G., Parker, D., Zou, X.: Verification and control of partially
  observable probabilistic systems. Real Time Syst.  \textbf{53}(3),  354--402
  (2017)

\bibitem{DBLP:conf/ijcai/PineauGT03}
Pineau, J., Gordon, G.J., Thrun, S.: Point-based value iteration: An anytime
  algorithm for {POMDP}s. In: {IJCAI}. pp. 1025--1032. Morgan Kaufmann (2003)

\bibitem{DBLP:conf/atva/QuatmannD0JK16}
Quatmann, T., Dehnert, C., Jansen, N., Junges, S., Katoen, J.P.: Parameter
  synthesis for {M}arkov models: Faster than ever. In: {ATVA}. {LNCS},
  vol.~9938 (2016)

\bibitem{ruder2016overview}
Ruder, S.: An overview of gradient descent optimization algorithms. arXiv
  preprint arXiv:1609.04747  (2016)

\bibitem{DBLP:books/lib/Rumelhart89}
Rumelhart, D.E.: Parallel Distributed Processing. {MIT} Press (1989)

\bibitem{DBLP:conf/qest/SalmaniK20}
Salmani, B., Katoen, J.P.: Bayesian inference by symbolic model checking. In:
  {QEST}. {LNCS}, vol. 12289, pp. 115--133. Springer (2020)

\bibitem{bnrepository}
Scutari, M.: Bayesian network repository (2021),
  \url{https://www.bnlearn.com/bnrepository/}

\bibitem{DBLP:conf/nips/SilverV10}
Silver, D., Veness, J.: Monte-{C}arlo planning in large {POMDP}s. In: {NIPS}.
  pp. 2164--2172. Curran Associates, Inc. (2010)

\bibitem{smith1997penalty}
Smith, A.E., Coit, D.W., Baeck, T., Fogel, D., Michalewicz, Z.: Penalty
  functions. Handbook of evolutionary computation  \textbf{97}(1), ~C5 (1997)

\bibitem{DBLP:journals/jair/SpaanV05}
Spaan, M.T.J., Vlassis, N.A.: Perseus: Randomized point-based value iteration
  for {POMDP}s. J. Artif. Intell. Res.  \textbf{24},  195--220 (2005)

\bibitem{DBLP:conf/atva/SpelJK19}
Spel, J., Junges, S., Katoen, J.P.: Are parametric {M}arkov chains monotonic?
  In: ATVA. LNCS, vol. 11781, pp. 479--496. Springer (2019)

\bibitem{DBLP:conf/icml/SutskeverMDH13}
Sutskever, I., Martens, J., Dahl, G.E., Hinton, G.E.: On the importance of
  initialization and momentum in deep learning. In: {ICML} {(3)}. {JMLR}
  Workshop and Conference Proceedings, vol.~28, pp. 1139--1147. JMLR.org (2013)

\bibitem{DBLP:books/daglib/0014221}
Thrun, S., Burgard, W., Fox, D.: Probabilistic Robotics. {MIT} Press (2005)

\bibitem{RMSProp}
Tieleman, T., Hinton, G.: {Lecture 6.5---{RMSProp}: Divide the gradient by a
  running average of its recent magnitude}. COURSERA: Neural Networks for
  Machine Learning (2012)

\bibitem{DBLP:books/daglib/0095265}
Vanderbei, R.J.: Linear programming - foundations and extensions, Kluwer
  {I}nternational {S}eries in {O}perations {R}esearch and {M}anagement
  {S}ervice, vol.~4. Kluwer (1998)

\bibitem{DBLP:conf/aaai/WalravenS17}
Walraven, E., Spaan, M.T.J.: Accelerated vector pruning for optimal {POMDP}
  solvers. In: {AAAI}. pp. 3672--3678. {AAAI} Press (2017)

\bibitem{DBLP:journals/tac/WintererJWJTKB21}
Winterer, L., Junges, S., Wimmer, R., Jansen, N., Topcu, U., Katoen, J.P.,
  Becker, B.: Strategy synthesis for {POMDP}s in robot planning via game-based
  abstractions. {IEEE} Trans. Autom. Control.  \textbf{66}(3),  1040--1054
  (2021)

\bibitem{DBLP:conf/nfm/Winterer00020}
Winterer, L., Wimmer, R., Jansen, N., Becker, B.: Strengthening deterministic
  policies for {POMDP}s. In: {NFM}. {LNCS}, vol. 12229, pp. 115--132. Springer
  (2020)

\end{thebibliography}
\iftoggle{TR}{
\clearpage
	\appendix
	\section{Proofs}
\defsoed*
First of all, we observe that an alternative way to write the equations in \cref{def:soed} for \(S \setminus \{\good\} = \{s_0,\ldots, s_{n-1}\}\) is:
\begin{align*}
	& \left(1 - \left(
	\begin{array}{c|c}
		A
		&
		0
		\\ \hline
		\pt{} {A}
		&
		A
		\\
	\end{array}
	\right) \right)
	\begin{pmatrix}
		\rv{x_0} \\ \vdots \\ \rv{x_{n-1}} \\
		\mbdervar{p}{x_0} \\ \vdots \\ \mbdervar{p}{x_{n-1}}
	\end{pmatrix}
	=
	\begin{pmatrix}
		\rewFunction(s_0) \\ 
		\vdots \\
		\rewFunction(s_{n-1}) \\ 
		0 \\ 
		\vdots \\
		0 \\ 
	\end{pmatrix}
\end{align*}
where $A$ equals the transition probability function $\probdtmc$ restricted to $S \setminus \{\good\}$.
\label{prf:thm:ds_unique}
\thmdsunique*
	\begin{proof}[of \cref{thm:ds_unique}]
		Clearly $x_s$ equals $\solRew{s}{\good}$.
		Furthermore, observe that in order to obtain the derivative for $x_s$, \ie $\mb{\pt{} x_s}$, we apply the sum rule and the product rule to $\rv{x_s}$. 
	For the variables \(\rv{x_0} , \hdots , \rv{x_{n-1}}\) we have:
	\begin{align*}
	(1 - A)
	\begin{pmatrix}
		\rv{x_0} \\ \vdots \\ \rv{x_{n-1}}
	\end{pmatrix}
	&=
	\begin{pmatrix}
		\rewFunction(s_0) \\ 
		\vdots \\
		\rewFunction(s_{n-1}) \\ 
	\end{pmatrix} \\
	\Leftrightarrow
	\begin{pmatrix}
		\rv{x_0} \\ \vdots \\ \rv{x_{n-1}}
	\end{pmatrix}
	&=
	(1 - A)^{-1}
	\begin{pmatrix}
		\rewFunction(s_0) \\ 
		\vdots \\
		\rewFunction(s_{n-1}) \\ 
	\end{pmatrix}
	=
	\begin{pmatrix}
		\solRew{s_0}{\good} \\
		\vdots \\
		\solRew{s_{n-1}}{\good}\\
	\end{pmatrix}
	\end{align*}
	\((1 - A)^{-1}\) exists because \(1-A\) is invertible (proof in \cite[p. 821]{BK08}).
	For the variables \(\mbdervar{p}{x_0}, \hdots , \mbdervar{p}{x_{n-1}}\) we have:
	\begin{align*}
		\begin{pmatrix}
			\mbdervar{p}{x_0} \\ \vdots \\ \mbdervar{p}{x_{n-1}}
		\end{pmatrix}
		=&
		(\dervar{p} A)
		\begin{pmatrix}
			\rv{x_0}\\ \vdots \\ \rv{x_{n-1}}
		\end{pmatrix}
		+
		A
		\begin{pmatrix}
			\mbdervar{p}{x_0} \\ \vdots \\ \mbdervar{p}{x_{n-1}}
		\end{pmatrix} \\
		\Longleftrightarrow
		(1-A)
		\begin{pmatrix}
			\mbdervar{p}{x_0} \\\vdots \\ \mbdervar{p}{x_{n-1}}
		\end{pmatrix}
		=&
		(\dervar{p} A)
		\begin{pmatrix}
		\solRew{s_0}{\good} \\
		\vdots \\
		\solRew{s_{n-1}}{\good}\\
		\end{pmatrix} \\
		\Longleftrightarrow
		(1-A)
		\begin{pmatrix}
			\mbdervar{p}{x_0} \\ \vdots \\ \mbdervar{p}{x_{n-1}}
		\end{pmatrix}
		=&
		(\dervar{p} A)
		\begin{pmatrix}
			\solRew{s_0}{\good} \\
			\vdots \\
			\solRew{s_{n-1}}{\good}\\
		\end{pmatrix} \\
		\Leftrightarrow
		\begin{pmatrix}
			\mbdervar{p}{x_0} \\ \vdots \\ \mbdervar{p}{x_{n-1}}
		\end{pmatrix}
		=&
		(1 - A)^{-1}
		\left(
		(\dervar{p} A)
		\begin{pmatrix}
			\solRew{s_0}{\good} \\
			\vdots \\
			\solRew{s_{n-1}}{\good}\\
		\end{pmatrix}
		\right)
		\\
	\end{align*}
	Thus there exists a unique solution of the system of equations.
\end{proof}

	\clearpage
}

\end{document}